\documentclass[journal]{IEEEtran}

\usepackage{amsmath}
\usepackage{amssymb}
\usepackage{algpseudocode}
\usepackage{algorithm}
\usepackage{algpseudocode}
\usepackage{float}
\newtheorem{lemma}{Lemma}
\newtheorem{theorem}{Theorem}
\usepackage{capt-of}
\usepackage{booktabs}
\usepackage{tabularx}
\usepackage{array}
\usepackage{cite}
\usepackage{graphicx}
\usepackage{dblfloatfix}
\usepackage{placeins}
\usepackage[hidelinks,colorlinks=true,linkcolor=blue,citecolor=blue,urlcolor=blue]{hyperref}
\usepackage{url}
\ifCLASSINFOpdf

\begin{document}

\title{GeoDose-CP: Graph-Local Conformal Inference for Continuous-Treatment Earth Observation}
\author{Md Khalid Hasan Sakib, Dristi Datta*, Manoranjan Paul, and Davina White%
\thanks{Md Khalid Hasan Sakib is with the Department of Computer Science and
Engineering, Uttara University, Dhaka, Bangladesh
(e-mail: 2261101010@uttara.ac.bd).}%
\thanks{Dristi Datta is with the School of Computing, Mathematics and
Engineering, Charles Sturt University, Bathurst, NSW 2795, Australia,
and also with Australian Integrated Carbon (AiCarbon), Level 4,
191 Pulteney Street, Adelaide, SA 5000, Australia
(e-mail: ddatta@csu.edu.au; dristi.datta@aicarbon.com).
Corresponding author: Dristi Datta.}%
\thanks{Manoranjan Paul is with the School of Computing, Mathematics and
Engineering, Charles Sturt University, Bathurst, NSW 2795, Australia
(e-mail: mpaul@csu.edu.au).}%
\thanks{Davina White is with Australian Integrated Carbon (AiCarbon),
Level 4, 191 Pulteney Street, Adelaide, SA 5000, Australia
(e-mail: davina.white@aicarbon.com).}%
}

\maketitle


\begin{abstract}
Reliable intervention-oriented uncertainty quantification from Earth observation (EO) remains difficult when continuous treatment shift, spatial dependence, limited support, and satellite-outcome uncertainty must be addressed simultaneously. Existing causal, conformal, and spatial approaches address components of this problem, but their direct combination does not generally recover the correct interventional reference law because candidate reassignment jointly alters treatment likelihood, standardized residuals, and graph-dependent residual likelihood. This study presents GeoDose-CP, a support-aware conformal framework for localized stochastic potential outcomes under continuous or mixed continuous--atomic treatment. Its central methodological contribution is a graph-local target-orbit law that jointly represents intervention-induced treatment shift, the inverse outcome-scale Jacobian, and spatial residual dependence, together with exact weighted candidate inversion, a scalable sparse approximation with explicit discrepancy accounting, and refusal under inadequate support. Evaluation used controlled known-truth experiments, MineDoseBench, treatment-density sensitivity analysis, external conformal comparators, and a multi-mine New South Wales (NSW) study. In MineDoseBench, GeoDose-CP achieved mean selective coverage of 0.9692 across 27 configurations and a minimum local $q_{0.05}$ of 0.8951; exact--sparse auditing produced nine inclusion disagreements over 2,700 targets. In the NSW study, the absence of an auditable longitudinal rehabilitation treatment rendered treatment-dependent inference nonoperational rather than forcing inference through a proxy exposure.
\end{abstract}


%
\IEEEpeerreviewmaketitle


\section{Introduction}
\label{sec:introduction}

Earth observation (EO) increasingly supports environmental management and
restoration assessment at spatial scales where controlled field experiments
are difficult or infeasible. Repeated satellite observations provide
spatially continuous measurements of land-surface change from regional to
global scales
\cite{GorelickEtAl2017EarthEngine,ReichsteinEtAl2019EarthSystemML},
creating opportunities to move beyond descriptive mapping toward inference
on outcomes under specified interventions. This distinction is particularly
important in mine rehabilitation, where satellite time series quantify
disturbance and vegetation recovery
\cite{WernerEtAl2020GlobalMineAreas,MausEtAl2022MiningLandUse,
TangWerner2023MiningFootprint,VidalMacuaEtAl2020MineRecovery}, but observed
recovery does not identify what would have occurred under a different
rehabilitation intensity. Intervention-oriented EO inference must therefore
address non-random treatment assignment, continuous and spatially
heterogeneous treatment support, spatial dependence, geographic deployment,
and uncertainty in satellite-derived outcomes.

EO research has developed powerful methods for spatial prediction and
validation. Gaussian-process models, machine learning, and Earth-system
learning represent complex nonlinear and spatial structure
\cite{CampsVallsEtAl2016GaussianProcessesEO,
SvendsenEtAl2020DeepGPEO,ReichsteinEtAl2019EarthSystemML}, while
geographically separated validation is increasingly recommended for spatial
deployment
\cite{RobertsEtAl2017SpatialCV,ValaviEtAl2019BlockCV,
MeyerEtAl2018TargetValidation}. Such validation can reveal substantially
weaker generalization than random partitioning
\cite{PlotonEtAl2020SpatialValidation}, and predictor selection and
transferability can depend strongly on deployment geography
\cite{MeyerEtAl2019SpatialPredictors,LudwigEtAl2023Transferability}.
These advances strengthen geographic prediction, but their target generally
remains an outcome under the observed data-generating regime rather than a
potential outcome under a specified intervention.

Causal inference addresses the intervention problem directly.
Generalized propensity-score methods extended treatment-effect analysis to
continuous treatment regimes
\cite{Imbens2000DoseResponse,ImaiVanDyk2004GeneralTreatment}, followed by
doubly robust, balancing, and locally adjusted dose--response estimators
\cite{KennedyEtAl2017ContinuousTreatment,FongEtAl2018CBPSContinuous,
PapadogeorgouDominici2020LocalConfounding}. These methods support causal
reasoning across a treatment continuum, but estimating a dose--response
function does not itself provide a conformal prediction set for a supported
potential outcome or determine how treatment transport should interact with
spatially dependent calibration.

Conformal prediction addresses uncertainty from a complementary direction.
Classical methods provide distribution-free predictive coverage under
exchangeability
\cite{ShaferVovk2008Tutorial,LeiWasserman2014PredictionBands,
LeiEtAl2018PredictiveInference,BarberEtAl2021JackknifePlus}, with later
extensions to counterfactual outcomes
\cite{LeiCandes2021Counterfactuals,QinEtAl2025CovariateShiftCausal},
continuous-treatment causal and dose--response inference
\cite{SchroederEtAl2025ContinuousTreatmentCP,
VerhaegheEtAl2026DoseResponseCP}, nonexchangeable data
\cite{BarberEtAl2023BeyondExchangeability,
OliveiraEtAl2024Nonexchangeable}, localized calibration
\cite{Guan2023LocalizedConformal,HoreBarber2025LocalWeights}, and covariate
or distribution shift
\cite{QiuEtAl2023CovariateShift,YangEtAl2024DoublyRobust,
GibbsCandes2024OnlineShift,ZhangEtAl2026TransferConformal}. Model-free
spatial conformal prediction has also been developed for geographically
dependent observations
\cite{MaoMartinReich2024SpatialPrediction}. These developments address
important parts of the EO problem, but they do not compose automatically:
localization does not define an interventional target law,
distribution-shift weighting does not encode graph-dependent residual
reassignment, and spatial conformal calibration does not identify a
continuous-treatment potential outcome.

The unresolved EO problem is therefore conformal potential-outcome inference
when \emph{intervention-induced treatment shift and spatially dependent
calibration act jointly}. Under candidate reassignment, both the
observational treatment likelihood and graph-residual likelihood change,
so the target reassignment probability cannot in general be constructed by
multiplying independent treatment and spatial weights. The required object
is a \emph{joint graph-local target law}.

We address this gap with \emph{GeoDose-CP}, a support-aware conformal
framework for localized stochastic potential outcomes under continuous or
mixed continuous--atomic treatment. Its central methodological contribution
is the graph-local target-orbit law (G2), which jointly represents
treatment shift, the inverse outcome-scale Jacobian, and graph-dependent
residual likelihood. GeoDose-CP connects this law to exact candidate
inversion, a scalable sparse approximation with explicit discrepancy
accounting, and refusal when treatment or spatial support is inadequate.
Causal identification and conformal validity remain separate requirements.

The principal contributions are:

\begin{itemize}

\item \textbf{Graph-local target law:}
a supported localized stochastic potential outcome and its corresponding
graph-local target-orbit law are derived for continuous or mixed
continuous--atomic treatment, establishing why independent
treatment--spatial factorization is not valid in general.

\item \textbf{Exact, scalable, and support-aware inference:}
the joint target law is coupled with exact weighted candidate inversion,
a target-weighted sparse approximation, and explicit refusal when treatment
support, endpoint support, graph connectivity, local calibration
information, or the applicable certification requirement is inadequate.

\item \textbf{Known-truth and practical evaluation:}
controlled experiments and MineDoseBench evaluate selective and local
coverage, treatment and target-population shift, spatial dependence,
measurement error, endpoint support, non-Gaussian dependence,
change of spatial support, exact--sparse approximation, refusal, and
coverage-matched efficiency. Secondary analyses examine treatment-density
estimation and compare GeoDose-CP with adapted continuous-treatment,
distribution-shift, and spatial/local conformal families.

\item \textbf{Real EO applicability:}
a multi-mine New South Wales (NSW) study evaluates geographically buffered
validation, EO-product quality, change of spatial support, and site/model
heterogeneity. Because the public archive does not provide an authentic
longitudinal rehabilitation treatment, treatment-dependent inference remains
nonoperational rather than being forced through a proxy exposure.

\end{itemize}

The evidence is intentionally complementary. Controlled experiments and
MineDoseBench provide known-truth evaluation of the complete
treatment--spatial architecture, whereas the NSW study tests whether real EO
data provide the treatment provenance, spatial support, and product quality
required for scientifically defensible inference. This separation evaluates
both inferential performance and whether an intervention query is
scientifically supportable.

The remainder of the paper develops the methodology
(Section~\ref{sec:method}), describes the evaluation design
(Section~\ref{sec:evaluation}), reports the results
(Section~\ref{sec:results}), discusses their implications and limitations
(Section~\ref{sec:discussion}), and concludes the study
(Section~\ref{sec:conclusion}).

\section{GeoDose-CP: Problem Formulation and Method}
\label{sec:method}

Figure~\ref{fig:architecture} summarizes GeoDose-CP from supported
intervention specification to graph-local reference-law construction,
exact or scalable conformal inference, and support-aware prediction or
refusal. The notation follows the methodological hierarchy developed
below: G1 is the finite-orbit conditional law, G2 the graph-local
target-orbit law, G3 exact weighted candidate inversion, N2 the sparse
graph approximation, and N3 the practical coverage-transfer layer.
Conditions C1--C5 define the causal-identification requirements.

\begin{figure*}[!t]
    \centering
    \includegraphics[width=\textwidth]{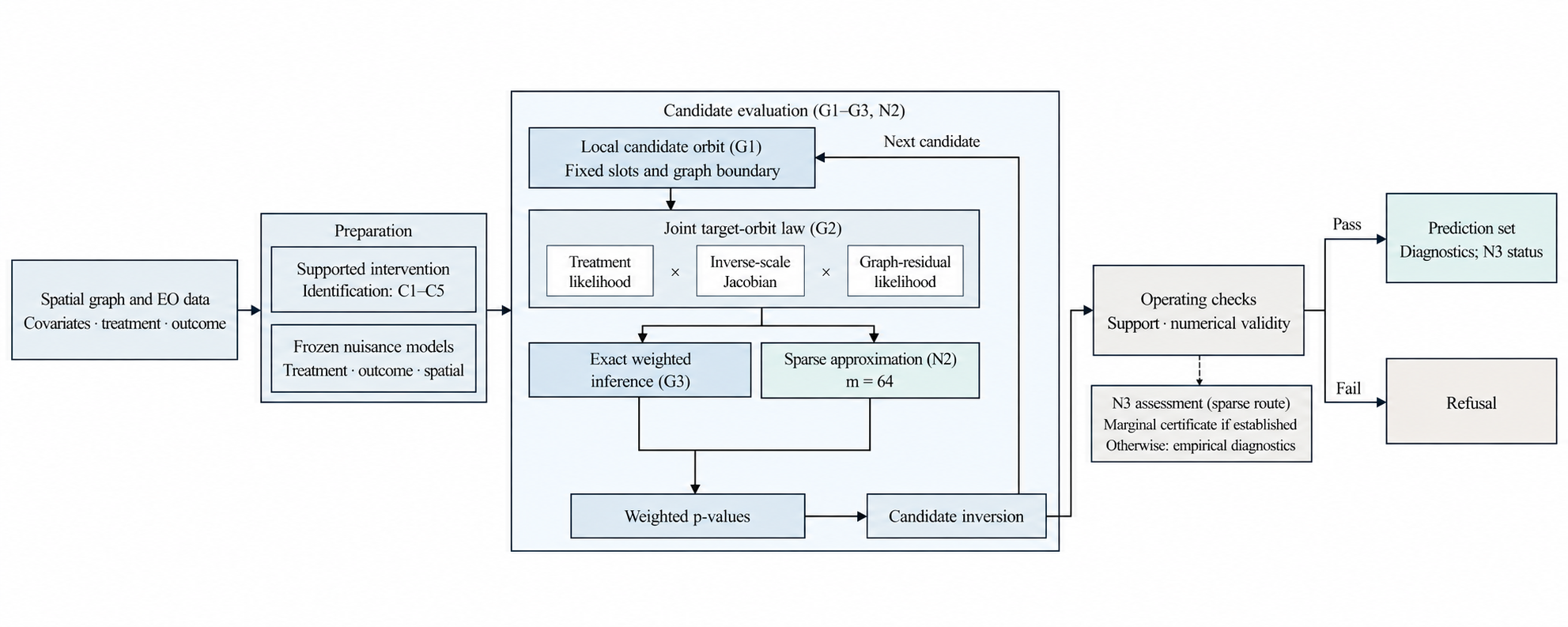}
    \caption{Operational workflow of GeoDose-CP. Spatial graph and EO data,
together with frozen nuisance models and the supported-intervention
conditions C1--C5, define the inputs to candidate evaluation. For each
candidate response, G1 constructs the local candidate orbit and G2 evaluates
the joint target-orbit law through the treatment likelihood, inverse-scale
Jacobian, and graph-residual likelihood. Inference then proceeds through
exact weighted inference (G3) or the sparse approximation (N2), followed by
weighted $p$-value calculation and candidate inversion. Support and numerical
checks, together with the applicable N3 assessment for the sparse route,
determine whether the query returns a prediction set with diagnostics or an
explicit refusal.}
    \label{fig:architecture}
\end{figure*}

\subsection{Supported Continuous-Dose Target and Identification}
\label{subsec:target}

Consider spatial units $i\in\mathcal{V}$ connected by a prespecified graph
$\mathcal{G}=(\mathcal{V},\mathcal{E})$. For each unit, $U_i$ denotes
pretreatment information, $A_i\in\mathcal{A}$ a continuous or mixed
continuous--atomic treatment, and $Y_i$ the observed outcome. The graph
encodes statistical spatial dependence for inference and does not itself
imply causal interference
\cite{MaoMartinReich2024SpatialPrediction,
HudgensHalloran2008Interference,
TchetgenVanderWeele2012Interference}.

GeoDose-CP targets a \emph{supported localized stochastic intervention}
around a scientifically specified dose $a_0$, rather than an unrestricted
deterministic potential outcome $Y(a_0)$. Let $\nu$ be a dominating
treatment measure containing Lebesgue mass on the continuous interior of
$\mathcal{A}$ and point masses only at genuine observational atoms. Both
the observational treatment law $g(a\mid u)$ and intervention law
$q_h(a\mid u;a_0)$ are defined with respect to this common mixed measure
$\nu$
\cite{Imbens2000DoseResponse,
ImaiVanDyk2004GeneralTreatment,
KennedyEtAl2017ContinuousTreatment}.

For a nonnegative localization kernel $\kappa_h$ with bandwidth $h>0$,

\begin{equation}
q_h(a\mid u;a_0)
=
\frac{
\kappa_h(a,a_0,u)
}{
\displaystyle
\int_{\mathcal{A}}
\kappa_h(s,a_0,u)\,d\nu(s)
}.
\label{eq:localized_intervention}
\end{equation}

The intervention is normalized with respect to the full mixed dominating
measure $\nu$. In the implemented analyses, its continuous component uses
Gaussian localization around $a_0$ on the supported interior $(0,1)$,
whereas exact endpoint queries can place mass only on genuine audited
observational atoms. Continuous treatment mass is never converted into
artificial endpoint mass.

The intervention is restricted to observational support:

\begin{equation}
q_h(a\mid u;a_0)>0
\quad\Longrightarrow\quad
g(a\mid u)>0.
\label{eq:positivity}
\end{equation}

Because \eqref{eq:positivity} is defined on the mixed measure $\nu$, an
endpoint is inferentially available only when it is a genuine observational
atom; unsupported endpoints cannot be justified by nearby continuous
support.

Let $\Pi^\star$ denote the intended deployment distribution and
$\chi(u,\mathcal{G})\in\{0,1\}$ an eligibility indicator determined from
pretreatment support, spatial design, and EO-quality information. Assuming
positive eligible mass,

\begin{equation}
d\Pi_{\mathrm{supp}}(u)
=
\frac{
\chi(u,\mathcal{G})\,d\Pi^\star(u)
}{
\displaystyle
\int
\chi(v,\mathcal{G})\,d\Pi^\star(v)
}.
\label{eq:supported_target}
\end{equation}

When transport from the observational population is required, assume
$\Pi_{\mathrm{supp}}\ll\Pi_{\mathrm{obs}}$ on the inferential support and
define

\begin{equation}
r(u)
=
\frac{
d\Pi_{\mathrm{supp}}
}{
d\Pi_{\mathrm{obs}}
}(u).
\label{eq:design_ratio}
\end{equation}

The target-design ratio $r(u)$ transports between populations and is
distinct from the treatment/intervention ratio $q_h/g$
\cite{QiuEtAl2023CovariateShift,
YangEtAl2024DoublyRobust,
QinEtAl2025CovariateShiftCausal}.

The supported target query is

\begin{equation}
U^\dagger\sim\Pi_{\mathrm{supp}},
\qquad
A^\dagger\sim q_h(\cdot\mid U^\dagger;a_0),
\qquad
Y^\dagger=Y^\dagger(A^\dagger).
\label{eq:target_generation}
\end{equation}

Let $P_0$ denote the resulting supported-target law and $\mathbb{E}_0$
expectation under $P_0$. Identification requires: \emph{C1} consistency,
$Y_i=Y_i(A_i)$; \emph{C2} supported conditional ignorability,
$Y_i(a)\perp A_i\mid U_i$ for
$\Pi_{\mathrm{supp}}$-almost every $u$ and
$q_h(\cdot\mid u;a_0)$-almost every supported $a$;
\emph{C3} positivity as in \eqref{eq:positivity};
\emph{C4} no material interference for the primary estimand; and
\emph{C5} stability of the relevant conditional outcome law between
observational and target regimes
\cite{KennedyEtAl2017ContinuousTreatment,
GiffinEtAl2023SpatialInterference,
PapadogeorgouEtAl2022SpatioTemporalCausal}.

Under C1--C5, for any integrable $\phi$,

\begin{equation}
\mathbb{E}_{0}
\!\left[
\phi(Y^\dagger)
\right]
=
\mathbb{E}_{\mathrm{obs}}
\left[
r(U)
\frac{
q_h(A\mid U;a_0)
}{
g(A\mid U)
}
\phi(Y)
\right].
\label{eq:identification}
\end{equation}

Equation~\eqref{eq:identification} identifies the supported stochastic
target law but does not establish conformal coverage. Causal identification
and conformal validity are therefore separate requirements
\cite{LeiCandes2021Counterfactuals,
QinEtAl2025CovariateShiftCausal}.

\subsection{Residual Representation and Failure of Naive Factorization}
\label{subsec:factorization}

Outcome nuisance functions are estimated from designated nuisance-training
data and held fixed during calibration. Let $\widehat{m}(u,a)$ and
$\widehat{\sigma}(u,a)>0$ denote the fitted location and scale functions
\cite{LeiEtAl2018PredictiveInference,
BarberEtAl2021JackknifePlus}. The standardized residual and primary
two-sided nonconformity score are

\begin{equation}
E_i
=
\frac{
Y_i-\widehat{m}(U_i,A_i)
}{
\widehat{\sigma}(U_i,A_i)
},
\qquad
S_i=|E_i|.
\label{eq:residual_score}
\end{equation}

The spatial dependence of
$E=(E_i:i\in\mathcal{V})$ is statistical and remains distinct from the
causal-interference condition C4
\cite{MaoMartinReich2024SpatialPrediction,
GiffinEtAl2023SpatialInterference}.

Suppose $w_{\mathrm{dose}}$ and $w_{\mathrm{spatial}}$ are constructed
separately to account for treatment shift and spatial calibration. A
seemingly natural combination is

\begin{equation}
w^{\mathrm{naive}}
\propto
w^{\mathrm{dose}}
w^{\mathrm{spatial}}.
\label{eq:naive_product}
\end{equation}

This factorization is not valid in general. Reassigning a
treatment--outcome payload changes the observational treatment likelihood,
its standardized residual through \eqref{eq:residual_score}, and the
graph-dependent residual likelihood simultaneously. The probability of a
candidate reassignment must therefore be evaluated under their joint law,
rather than reconstructed from independently derived marginal weights.

Product factorization is recovered only under additional structure that
separates the treatment and spatial contributions compatibly. Such
reductions are useful special cases, but good empirical performance of a
product-form procedure does not establish it as the correct reference law.
More generally, separate treatment and spatial marginals are insufficient
to recover the target--calibration coupling without the required graph-local
dependence structure. This motivates the joint graph-local target-orbit law
derived next.

\subsection{Exact Graph-Local Target-Orbit Law}
\label{subsec:target_orbit}

\subsubsection{Finite Local Orbit}

For target slot $t$, define the graph-local block

\begin{equation}
B=\{b_1,\ldots,b_{m_B},t\},
\label{eq:local_block}
\end{equation}

where $m_B$ is the number of calibration slots. Slot-specific
pretreatment information, graph locations, and design quantities remain
fixed, while treatment--outcome payloads are reassigned within $B$.

For candidate response $y$, the target outcome is first replaced by $y$,
after which permutations producing identical physical states are
quotiented. Let $\Omega_y$ denote the resulting set of distinct orbit
states. Let $\mathcal{I}_B$ denote the complete finite-orbit conditioning
information comprising the fixed slot and design information, payloads
outside $B$, and the unordered multiset of movable payloads within $B$.

\begin{lemma}[Finite-Orbit Conditional Law (G1)]
\label{lem:finite_orbit}
Assume the block payload law admits a density with respect to a product
dominating measure invariant under permutations of the block coordinates.
Conditional on $\mathcal{I}_B$, the probability of each distinct state
$\pi\in\Omega_y$ is proportional to its joint reference-law density.
Duplicate physical states are represented once; their stabilizer
multiplicities cancel after normalization.
\end{lemma}

Thus, conditional inference reduces to comparison over a finite set of
distinct graph-local payload assignments.

\subsubsection{Graph-Local Target Law}

For $\pi\in\Omega_y$, let
$(A_{\pi(i)},Y_{\pi(i)}(y))$ denote the payload assigned to slot $i$ and
define

\begin{equation}
E_i^\pi(y)
=
\frac{
Y_{\pi(i)}(y)
-
\widehat{m}\!\left(U_i,A_{\pi(i)}\right)
}{
\widehat{\sigma}\!\left(U_i,A_{\pi(i)}\right)
}.
\label{eq:state_residual}
\end{equation}

Let $\mathcal{C}$ denote the residual-graph clique collection and
$\psi_C$ the potential associated with clique $C$.

\begin{theorem}[Graph-Local Target-Orbit Law (G2)]
\label{thm:graph_local}
Assume the supported mixed-measure treatment law of
Section~\ref{subsec:target}, with treatment densities defined with respect
to the common mixed measure $\nu$, a slot-factorized observational
treatment law, and a positive standardized-residual density that, conditional on
the frozen pretreatment and design information, is invariant to
treatment after the specified location--scale transformation and
factorizes over the prespecified graph. Conditional on $\mathcal{I}_B$, the target weight of
$\pi\in\Omega_y$ satisfies

\begin{equation}
\begin{aligned}
W_\pi(y)
\propto\;&
q_h\!\left(A_{\pi(t)}\mid U_t;a_0\right)
\prod_{i\in B\setminus\{t\}}
g\!\left(A_{\pi(i)}\mid U_i\right)
\\
&\times
\prod_{i\in B}
\frac{1}{
\widehat{\sigma}\!\left(U_i,A_{\pi(i)}\right)}
\prod_{\substack{
C\in\mathcal{C}:\\
C\cap B\neq\varnothing
}}
\psi_C\!\left(E_C^\pi(y)\right).
\end{aligned}
\label{eq:graph_local_weight}
\end{equation}

Residual factors supported entirely outside $B$ are orbit-invariant and
cancel after normalization; hence only the moved block and its incident
graph boundary contribute to the spatial term.
\end{theorem}

Equation~\eqref{eq:graph_local_weight} is the central GeoDose-CP
reference law. It couples three quantities that generally cannot be
constructed independently: the intervention/observational treatment
likelihood, the inverse outcome-scale Jacobian, and the graph-local
residual likelihood. The Jacobian arises from the transformation
$Y\mapsto E$ in \eqref{eq:state_residual} and is required whenever
$\widehat{\sigma}$ changes across orbit states.

The target-design factor $r(U_t)$ is constant within the conditioned
orbit and therefore cancels from normalized orbit probabilities, although
it remains relevant for target-population transport and scalable coverage
accounting. Likewise, the exact treatment term requires $g(A\mid U)$,
or a theorem-equivalent joint assignment likelihood; a
target-versus-observational classifier estimates a different density
ratio and cannot replace $g$ in the exact target-orbit law. Dependent
treatment assignment would instead require the corresponding joint
assignment likelihood.

\noindent\emph{Proof sketch:}
By Lemma~\ref{lem:finite_orbit}, each distinct orbit state is weighted by
its joint reference-law density. Re-expressing moved outcomes as
standardized residuals contributes the inverse-scale Jacobian, while graph
factorization expresses the residual density through clique potentials.
Factors unaffected by the reassignment cancel after normalization,
leaving exactly the treatment, Jacobian, and boundary-touching graph terms
in \eqref{eq:graph_local_weight}. Complete proofs of G1--G3 are provided
in Supplementary Appendix~A.

\subsubsection{Gaussian Markov Specialization}

For a proper Gaussian Markov random field with sparse precision matrix
$Q$, partition the residual vector into the moved block $B$ and its
required external boundary $D$. Up to orbit-invariant terms,

\begin{equation}
\log L_{\mathcal{G}}^\pi(y)
=
-\frac{1}{2}
\left(E_B^\pi(y)\right)^{\mathsf T}
Q_{BB}E_B^\pi(y)
-
\left(E_B^\pi(y)\right)^{\mathsf T}
Q_{BD}E_D.
\label{eq:gmrf_local}
\end{equation}

The Gaussian normalizing term is orbit-invariant and therefore cancels.
Only sparse operations involving $B$ and its boundary are required, so
inversion of the full spatial covariance matrix is unnecessary.
Equation~\eqref{eq:gmrf_local} is a computational specialization of
Theorem~\ref{thm:graph_local}; the G2 result itself is non-Gaussian and
requires only the stated positive graph-factorized residual law.

\subsection{Exact Candidate Inversion and Scalable Coverage}
\label{subsec:coverage}

\subsubsection{Exact Weighted Candidate Inversion}

Let $W_\pi(y)$ be the unnormalized target-orbit weight in
\eqref{eq:graph_local_weight}. Define

\begin{equation}
\overline{W}_\pi(y)
=
\frac{
W_\pi(y)
}{
\displaystyle
\sum_{\pi'\in\Omega_y}
W_{\pi'}(y)
}.
\label{eq:normalized_weight}
\end{equation}

Let $T_\pi(y)$ denote the target-slot nonconformity score under state
$\pi$, with the primary score induced by \eqref{eq:residual_score}, and
let $\pi_{\mathrm{obs}}$ denote the realized assignment. The weighted
upper-tail p-value is

\begin{equation}
p_y
=
\sum_{\pi\in\Omega_y}
\overline{W}_\pi(y)
\mathbf{1}
\left\{
T_\pi(y)
\geq
T_{\pi_{\mathrm{obs}}}(y)
\right\}.
\label{eq:weighted_pvalue}
\end{equation}

Including structural ties in the upper tail gives, under the exact
target-orbit law,

\begin{equation}
\Pr_{0}
\left\{
p_{Y^\dagger}\leq\alpha
\mid
\mathcal{I}_B
\right\}
\leq
\alpha,
\label{eq:exact_validity}
\end{equation}

where $\mathcal{I}_B$ is the finite-orbit conditioning information defined
in Section~\ref{subsec:target_orbit}. Candidate inversion yields

\begin{equation}
\widehat{C}_\alpha
=
\left\{
y\in\mathcal{Y}:p_y>\alpha
\right\}.
\label{eq:prediction_set}
\end{equation}

The raw set may be disconnected; all components are retained. Any convex
hull shown for visualization is presentation-only and does not replace the
set in \eqref{eq:prediction_set}.

Equations~\eqref{eq:weighted_pvalue}--\eqref{eq:prediction_set} define
G3. Once the correct target-orbit reference law is established, candidate
inversion follows generalized conformal validity machinery; the
GeoDose-specific structural contribution is therefore G2 in
Theorem~\ref{thm:graph_local}
\cite{BarberTibshirani2026Unifying,
PrinsterEtAl2024AnyDistribution}.

\subsubsection{Sparse Graph Approximation}

Exact orbit evaluation is feasible for small local blocks. For larger
domains, GeoDose-CP uses a Vecchia-type sparse conditional approximation
\cite{KatzfussGuinness2021Vecchia}. Under a deterministic ordering, let
$P_i$ be the full predecessor set of residual $E_i$ and
$N_i\subseteq P_i$ the retained neighborhood:

\begin{equation}
p(E)
=
\prod_i p(E_i\mid E_{P_i}),
\qquad
q(E)
=
\prod_i p(E_i\mid E_{N_i}).
\label{eq:sparse_factorization}
\end{equation}

The information omitted by truncation is

\begin{equation}
D_{\mathrm{KL}}(p\|q)
=
\sum_i
I
\left(
E_i;
E_{P_i\setminus N_i}
\mid
E_{N_i}
\right).
\label{eq:omitted_information}
\end{equation}

Let $L(U)$ denote the corresponding target-local information loss.
Assuming finite target-weighted information,

\begin{equation}
\mathbb{E}_{\mathrm{obs}}
\left[
r(U)L(U)
\right]
=
\mathbb{E}_{\Pi_{\mathrm{supp}}}
\left[
L(U)
\right]
<\infty,
\label{eq:sparse_integrability}
\end{equation}

and define

\begin{equation}
\Delta_{\mathrm{sp}}
=
\mathbb{E}_{\Pi_{\mathrm{supp}}}
\!\left[
L(U)
\right]
=
\mathbb{E}_{\mathrm{obs}}
\!\left[
r(U)L(U)
\right].
\label{eq:sparse_discrepancy}
\end{equation}

Pinsker's inequality then gives

\begin{equation}
\delta_{\mathrm{sparse}}
\leq
\sqrt{
\frac{\Delta_{\mathrm{sp}}}{2}
}.
\label{eq:pinsker}
\end{equation}

This target-weighted sparse approximation and its associated
coverage-transfer step constitute N2.

\subsubsection{Practical Marginal Coverage}

The implemented reference law may additionally differ from the ideal
supported-target law through structural misspecification, nuisance
estimation, and numerical approximation. Let
$\delta_{\mathrm{mis}}$, $\delta_{\mathrm{est}}$, and
$\delta_{\mathrm{comp}}$ denote valid total-variation bounds for these
transitions. KL divergence controls the sparse approximation through
Pinsker's inequality; subsequent law discrepancies are combined on the
total-variation scale.

The resulting law ladder is

\begin{equation}
P_0
\longrightarrow
P_{\mathrm{sparse}}
\longrightarrow
P_{\mathrm{struct}}
\longrightarrow
P_{\mathrm{fit}}
\longrightarrow
P_{\mathrm{comp}},
\label{eq:law_ladder}
\end{equation}

with each discrepancy charged once. If the joint nuisance/certificate
good event fails with probability at most $\eta$, then

\begin{equation}
\begin{aligned}
\Pr_{0}
\left\{
Y^\dagger \in \widehat{C}_{\alpha}
\right\}
\geq\;&
1-\alpha
-\delta_{\mathrm{sparse}}
-\delta_{\mathrm{mis}}
\\
&-\delta_{\mathrm{est}}
-\delta_{\mathrm{comp}}
-\eta .
\end{aligned}
\label{eq:practical_coverage}
\end{equation}

Equation~\eqref{eq:practical_coverage} defines N3 and is a
\emph{marginal} coverage-transfer statement. We use the term
\emph{N3 certificate} only when the displayed discrepancy bounds and
good-event probability are formally established. When fitted treatment,
outcome, spatial, or computational components lack the corresponding
finite-sample controls, the resulting N3 quantities are reported as
empirical diagnostics rather than theorem-certified coverage guarantees.
Conditional coverage would require corresponding conditional
total-variation control
\cite{GibbsCherianCandes2025Conditional}. For the exact route,
$\delta_{\mathrm{sparse}}=0$; for the scalable route it is controlled by
\eqref{eq:sparse_discrepancy}--\eqref{eq:pinsker}.
\subsection{Support-Aware Refusal}
\label{subsec:refusal}

GeoDose-CP returns either a supported prediction set or an explicit
refusal when the required target law cannot be justified. Applicable
gates assess intervention positivity, genuine endpoint support,
effective-sample-size and weight-concentration diagnostics, graph-safe
calibration support, connected local geometry, EO-quality eligibility,
and nonvacuity of the applicable N3 certificate or prespecified empirical
diagnostic criterion.

Accordingly,

\begin{equation}
\text{GeoDose-CP output}
\in
\left\{
\widehat{C}_\alpha,
\;
\mathrm{refusal}
\right\}.
\label{eq:refusal_output}
\end{equation}

Refusal is an inferential outcome rather than missing data or noncoverage.
Unsupported treatment values are not extrapolated, unsupported endpoint
atoms are not created from nearby continuous mass, and treatment-dependent
queries are not constructed when the scientifically required treatment
variable is unavailable. Effective sample size, weight concentration,
graph support, and endpoint checks remain operational support diagnostics
rather than general coverage guarantees. The method-specific operating
thresholds used in evaluation are reported in
Section~\ref{sec:evaluation}.

\subsection{GeoDose-CP Inference Algorithm}
\label{subsec:algorithm}

Let $\mathcal{F}$ collect the fitted nuisance models, graph, calibration
data, support rules, intervention specification, numerical settings, and
prespecified inferential route used for a query. For scalable inference,
$\mathcal{F}$ also records whether the applicable N3 assessment is a
formally established certificate or a prespecified empirical diagnostic.
Algorithm~\ref{alg:geodose} summarizes GeoDose-CP inference.

\begin{algorithm}[!t]
\caption{GeoDose-CP Query Inference}
\label{alg:geodose}
\footnotesize
\begin{algorithmic}[1]

\Require Target $t$, dose $a_0$, level $\alpha$, domain $\mathcal{Y}$,
and query objects $\mathcal{F}$

\If{an applicable eligibility or graph-support gate fails}
    \State \Return \textsc{Refuse}
\EndIf

\State Construct $q_h(\cdot\mid U_t;a_0)$ and local block $B_t$

\If{an applicable positivity, endpoint, or operating-support gate fails}
    \State \Return \textsc{Refuse}
\EndIf

\For{each candidate $y$ in the prespecified inversion routine}

    \State Replace the target response by $y$, then construct $\Omega_y$

    \State Compute candidate residuals using
    \eqref{eq:state_residual}

    \If{exact route}
        \State Evaluate target-orbit weights using
        \eqref{eq:graph_local_weight}
    \Else
        \State Evaluate the N2 sparse approximation
    \EndIf

    \State Normalize weights and compute $p_y$ using
    \eqref{eq:weighted_pvalue}

\EndFor

\If{a required numerical or computational check fails}
    \State \Return \textsc{Refuse}
\EndIf

\State Construct
$\widehat{C}_{\alpha}
=
\{y\in\mathcal{Y}:p_y>\alpha\}$

\If{sparse route}
    \State Assess the applicable N3 certificate or empirical diagnostic
    \If{the prespecified N3 operating criterion fails or is vacuous}
        \State \Return \textsc{Refuse}
    \EndIf
\EndIf

\State \Return $\widehat{C}_{\alpha}$, support diagnostics, and N3 status

\end{algorithmic}
\end{algorithm}

Algorithm~\ref{alg:geodose} does not require a particular outcome-prediction
architecture once the fitted location and scale functions are learned from
the designated nuisance-training data and held fixed during calibration.
In the empirical analyses, RF/XGBoost outcome models and fitted treatment
and spatial nuisance models instantiate these components. Their associated
N3 quantities are interpreted as empirical diagnostics unless the
finite-sample discrepancy controls required by
\eqref{eq:practical_coverage} are separately established.

\section{Data and Evaluation Design}
\label{sec:evaluation}

\subsection{Evidence Architecture and Data Substrate}
\label{subsec:data_substrate}

The evaluation uses three complementary evidence layers: controlled
experiments test the inferential mechanism under known treatment,
counterfactual, and spatial laws; MineDoseBench combines real mine geometry
and pretreatment environmental context with generated counterfactual truth;
and the New South Wales (NSW) study examines real EO applicability where
counterfactual truth is unavailable. Together, they address controlled
validity, realistic known-truth performance, and real-data applicability.

A common spatial substrate was constructed for Bulga Complex, Hunter Valley
Operations, Liddell Coal, Mt Arthur Coal, and Mount Thorley Warkworth
Complex from the NSW mine-rehabilitation archive
\cite{NSWResources2024MineRehabilitation}. Rehabilitation and disturbance
polygons defined the analysis footprint but were not treated as the
longitudinal rehabilitation exposure required by the causal estimand. A
globally anchored $90\,\mathrm{m}\times90\,\mathrm{m}$ lattice in
EPSG:9473 retained cells with at least 70\% supported area; queen
contiguity yielded 27,042 cells and 102,893 undirected within-mine graph
edges.

EO variables were derived from Digital Earth Australia (DEA) Landsat
Fractional Cover Collection 3 and DEA Water Observations
\cite{Lymburner2021DEAFractionalCover,
GeoscienceAustralia2022DEAWaterObservations}. Photosynthetic vegetation
(PV), unmixing error (UE), and water status were extracted for 2023--2025.
Acquisitions required at least 70\% valid and dry supported pixels, were
composited by annual pixel median, and were summarized within cells by
spatial median. Annual eligibility required at least three accepted
observations, at least 70\% valid composite pixels, finite PV and UE, water
fraction $\leq0.50$, and no more than 50\% of valid pixels with
UE$\geq25$. The resulting archive contains 81,126 block--year records, with
23,610 cells satisfying the complete three-year EO-quality criteria.
Mt Arthur Coal, Hunter Valley Operations, and Bulga Complex were selected
for the real-data study using this quality screen without reference to PV
level or PV change.

MineDoseBench adds 13 pretreatment environmental variables: four climate,
five soil, and four terrain attributes. Climate variables comprise annual
rainfall and short-crop evapotranspiration for 2023 and 2024 from SILO
\cite{QueenslandGovernmentSILO,JeffreyEtAl2001SILO}. Soil variables comprise
topsoil organic carbon, CaCl$_2$ pH, clay fraction, bulk density, and
available volumetric water capacity
\cite{WadouxEtAl2024SLGAOrganicCarbon,
MaloneSearle2024SLGApH,
MaloneSearle2022SLGAClay,
Malone2023SLGABulkDensity,
SearleEtAl2023SLGAAWC}. Terrain variables comprise elevation, slope, aspect,
and surface roughness derived from Australian SRTM elevation products
\cite{GallantEtAl2011SRTMDEM}. Raster quantities were aggregated by polygon
area weighting, and bulk density used the authoritative Release-2 unit of
$\mathrm{g\,cm^{-3}}$
\cite{Malone2023SLGABulkDensity}.

Environmental-context completeness was assessed independently of EO
eligibility. Of 23,710 benchmark cells requiring context information, the
prespecified 0.99 valid-area criterion retained 23,618 and excluded 92,
with no missing-value imputation. The 23,618 context-complete cells and
23,610 three-year EO-quality cells therefore represent distinct eligibility
screens.

Environmental context and pretreatment EO/history variables formed 26
permitted MineDoseBench nuisance predictors. Outcome-year and
outcome-derived variables, oracle or hidden states, sample-role indicators,
and mapped rehabilitation variables were excluded. Nuisance training,
support construction, calibration, and final testing were spatially
separated, with scenario definitions, targets, and operating rules specified
before performance evaluation.
\subsection{Controlled Validation and MineDoseBench}
\label{subsec:controlled_mdb}

Controlled known-truth experiments first isolate the inferential mechanism
before introducing realistic mine geometry and environmental heterogeneity.
The suite comprises 3,000 primary replications, 600 fixed-resolution
increasing-domain evaluations, and 590 stress evaluations, totaling 4,190
case replications and 52,740 scientific queries. Conditions vary treatment
and target-population shift, spatial dependence, overlap, residual structure,
graph information, measurement error, endpoint support, and calibration
information.

The increasing-domain experiment uses $17\times17$, $25\times25$, and
$35\times35$ lattices at fixed 90-m resolution, varying spatial extent
rather than analytical support. Candidate inversion is restricted to the
numerical domain $[-8,8]$, which is a computational range rather than a
physical fractional-cover domain.

MineDoseBench provides the primary realistic known-truth evaluation. It
retains NSW mine geometry, graph structure, pretreatment EO history, and
environmental context while generating treatment, potential outcomes,
spatial residual dependence, and target-population shift under an auditable
data-generating process. The benchmark response is fractional-cover change
on the physical domain

\begin{equation}
\mathcal{Y}_{\mathrm{MDB}}=[-1,1],
\label{eq:mdb_domain}
\end{equation}

with truth generated directly within this range rather than imposed by
post-hoc clipping.

MineDoseBench is constructed around the joint treatment--spatial problem:
mixed continuous--atomic treatment with genuine endpoint atoms, nonlinear
dose response and effect modification, graph-dependent residual variation,
and independently controlled target-population shift. Baseline
configurations satisfy the stated identification and reference-law
conditions, whereas stress configurations perturb overlap, treatment and
endpoint support, spatial and residual structure, calibration information,
measurement quality, and analytical support. Hidden spatial confounding is
included solely as a negative control and does not satisfy the
conditional-ignorability requirement for causal interpretation.

The treatment mechanism is known by construction in the controlled
experiments. In the registered MineDoseBench RF-primary pipeline, however,
the observational treatment law $g(A\mid U)$ is estimated from designated
nuisance-training data using the frozen mixed-propensity model
$\widehat g(A\mid U)$. The known MineDoseBench treatment law is retained
only for the separate treatment-density sensitivity analysis in
Section~\ref{subsec:sensitivity_comparators}, where oracle, estimated, and
deliberately misspecified treatment models are compared while the remaining
pipeline components are held fixed. This separation allows structural
behavior to be examined under known treatment mechanisms while evaluating
the primary realistic benchmark with a fitted treatment model.

MineDoseBench contains 27 configurations with 20 independent replications
each (540 tasks) and five ranked target cells in each of five mines per
configuration and replication, yielding 13,500 primary target identities.
Change of spatial support is evaluated on an independently anchored 180-m
substrate containing 7,295 cells, of which 5,788 satisfy the eligibility
criteria; retained 180-m cells require the corresponding available 90-m
children to be eligible. Additional configurations assess non-Gaussian
dependence, measurement error, reduced information, endpoint support,
combined perturbations, and cross-mine transfer. The three-year temporal
fixture is used only as a structural stress test and is not interpreted as
evidence of temporal-conformal validity.
\subsection{Methods, Operating Rules, and Evaluation Metrics}
\label{subsec:metrics}

The primary MineDoseBench comparison comprises five mechanism-aligned
procedures: dose-only weighted conformal prediction (M2), spatial-only
residual calibration (M3), naive treatment--spatial factorization (M4),
the treatment-weighted graph-safe baseline (M5), and full GeoDose-CP (M6).
M5 is a deterministic target-separated graph-safe weighted conformal
procedure using treatment weights on a pre-outcome eligible graph-safe
calibration subset. Ordinary split conformal prediction (M1) is retained
where informative but is not part of the primary M2--M6 comparison.

The spatial comparator M3-CRE reassigns standardized residual payloads
$E$ while holding treatment $A$, pretreatment information $U$, and graph
slots fixed, thereby isolating spatial recalibration without treatment
transport. M4 multiplies the M2 treatment marginal by the corresponding
M3 spatial marginal and normalizes once, providing the deliberately
factorized comparator to the joint graph-local construction in M6. Thus,
M2--M6 separate treatment shift, spatial recalibration, naive
factorization, graph-safe weighted calibration, and the joint
treatment--spatial target law. Published continuous-treatment,
distribution-shift, and spatial/local conformal families are evaluated
separately below.

Random forest (RF) is the primary nuisance-model track across all 27
configurations. XGBoost is a prespecified confirmation track restricted to
the S1/S4/S5/S8 subset and is not treated as a second complete benchmark
analysis. For treatment-dependent procedures in the RF-primary benchmark,
the observational treatment law is the fitted mixed-propensity model
$\widehat g(A\mid U)$ described in
Section~\ref{subsec:controlled_mdb}; the known oracle law is used only in
the separate treatment-density sensitivity analysis. The inferential level
and localization bandwidth are

\begin{equation}
\alpha=0.10,
\qquad
h=0.10.
\label{eq:primary_settings}
\end{equation}

The common numerical operating thresholds are

\begin{equation}
\mathrm{ESS}\geq3,
\qquad
\max_i \overline{w}_i\leq0.50,
\qquad
n_{\mathrm{GS}}\geq20,
\label{eq:operating_gates}
\end{equation}

where $\overline{w}_i$ is the applicable normalized operating weight and
$n_{\mathrm{GS}}$ is the graph-safe calibration-support count. These
thresholds were selected in a separate 20-replication pilot and frozen
before the independent controlled-production and MineDoseBench
evaluations; production outcomes were not used to retune them. They are
method-specific stability diagnostics rather than universal coverage
conditions. M2 uses treatment-weight ESS and maximum-weight conditions;
M3 uses graph-safe support without treatment-weight refusal; M4 and M5
apply their applicable weight and graph-support conditions; and M6 combines
the applicable weight and graph-support gates with a finite, positive N3
operating requirement. Positivity, endpoint, structural-support, and
computational checks are likewise applied only where defined. ESS and
weight concentration therefore diagnose support and stability rather than
establish coverage or treatment-model correctness by themselves.

The primary scalable M6 route uses the prospectively frozen sparse
neighborhood size $m=64$; additional neighborhood-size sensitivity is
reported in the Supplement. RF/XGBoost outcome models and fitted treatment
and spatial nuisance models provide practical implementations. In the
primary MineDoseBench implementation, the standardized outcome-residual
transformation uses $\widehat{\sigma}(U,A)\equiv1$; consequently, the
inverse-scale Jacobian in the G2 reference law is retained structurally
but equals one numerically in this benchmark. The spatial residual model
retains its separately fitted dispersion parameter. Associated N3
quantities are interpreted as empirical diagnostics unless the
finite-sample discrepancy controls required by (23) are separately
established.

The exact--sparse audit uses size-6 local blocks containing one target and
five distinct calibration slots. It contains 2,700 target identities and
8,100 method--target combinations across M3, M4, and M6; evaluating each
under both the full-graph exact reference and the sparse $m=64$ reference
yields 16,200 audit rows. Full finite-domain inversion uses 50 target
identities spanning
$\rho\in\{0,0.2,0.4,0.6,0.8\}$, with 10 identities per dependence
level. Evaluation across M3, M4, and M6 gives 150 full-inversion records.
For the corresponding efficiency-only comparison, M2 and M5 contribute
their closed-form returned intervals on the same registered target
identities, intersected with the benchmark response domain
$\mathcal{Y}_{\mathrm{MDB}}=[-1,1]$. The resulting width comparisons
therefore constitute a targeted structural-efficiency audit rather than
population-level efficiency evidence.

For query $q$, let $R_q=1$ indicate a returned prediction set and
$R_q=0$ a refusal. Selective coverage and return rate are

\begin{equation}
\widehat{\mathrm{Cov}}_{\mathrm{sel}}
=
\frac{
\sum_{q=1}^{Q}
R_q
\mathbf{1}
\left\{
Y_q^\dagger\in\widehat C_q
\right\}
}{
\sum_{q=1}^{Q}R_q
},
\qquad
\widehat r
=
\frac{1}{Q}
\sum_{q=1}^{Q}R_q.
\label{eq:selective_coverage}
\end{equation}

Refusals are reported separately and are not recoded as noncoverage.
Evaluation focuses on selective coverage, return/refusal, target-local
coverage, coverage-matched width, exact--sparse discrepancy, endpoint
behavior, and false-support behavior, where false support denotes returned
inference under a known benchmark support failure. For each prespecified
local reporting stratum, selective coverage is computed among returned
queries using the same convention as above; strata with no returned sets
have undefined selective coverage and remain represented through
return/refusal rather than being treated as coverage failures. Local
robustness is summarized by $q_{0.05}$, the fifth percentile of the
selective-coverage values across the prespecified mine and spatial-rank
cells. Dose-bin summaries are retained separately and are not included in
$q_{0.05}$. This quantity is a lower-tail diagnostic for concentrated
local failure and is not interpreted as exact conditional coverage
\cite{Guan2023LocalizedConformal,
HoreBarber2025LocalWeights,
GibbsCherianCandes2025Conditional}.

Configuration-level coverage uncertainty is summarized by 95\%
cluster-bootstrap intervals using replication-by-mine clusters as the
resampling units. Thus, the five target queries within a given
replication--mine cluster are kept together rather than treated as
independent bootstrap units. The lower endpoint of the resulting interval
is used as the reported configuration-level bootstrap lower bound. The
controlled study additionally evaluates 17 prespecified lower-tail coverage
groups using Bonferroni-adjusted undercoverage tests.

Prediction-set width is interpreted as an efficiency comparison only when

\begin{equation}
\left|
\widehat{\mathrm{Cov}}_{\mathrm{sel},j}
-
\widehat{\mathrm{Cov}}_{\mathrm{sel},k}
\right|
\leq0.03.
\label{eq:matched_coverage}
\end{equation}

Accordingly, a narrower procedure is credited as more efficient only under
comparable selective coverage; coverage, refusal, and sharpness remain
separate performance dimensions.

\subsection{Treatment-Density Sensitivity and External Comparator Evaluation}
\label{subsec:sensitivity_comparators}

Two secondary analyses examine treatment-law specification and external
conformal comparators without altering the registered
27-configuration MineDoseBench evaluation.

Treatment-density sensitivity is assessed within the RF-primary M6 pipeline
using

\begin{equation}
G_{\mathrm{ORACLE}},\qquad
G_{\mathrm{ESTIMATED}},\qquad
G_{\mathrm{MISSPECIFIED}}.
\label{eq:g_sensitivity_tracks}
\end{equation}

$G_{\mathrm{ORACLE}}$ uses the known mixed continuous--atomic treatment law
from the MineDoseBench data-generating process.
$G_{\mathrm{ESTIMATED}}$ uses the primary fitted propensity law
$\widehat g(A\mid U)$, estimated exclusively from nuisance-training data:
endpoint/interior mass is modeled by multinomial logistic regression, while
the interior conditional law uses a ridge-logit mean with a global beta
concentration parameter. $G_{\mathrm{MISSPECIFIED}}$ deliberately removes
dependence on $U$ through a marginal mixed-treatment model.

The analysis spans seven mechanism-based configurations: baseline,
treatment shift, strong spatial dependence ($\rho=0.6$), strong dependence
with target-design shift, poor overlap, severe-tail support stress, and
mixed endpoint atoms. All 20 replications and common target identities are
retained. The RF outcome nuisance, spatial nuisance, graph, calibration
data, intervention specification, support thresholds, $\alpha=0.10$, and
scalable inference route are held fixed, so the comparison isolates
treatment-law specification rather than changes elsewhere in the pipeline.

GeoDose-CP is also compared on a selected seven-configuration
MineDoseBench subset, wherever native applicability permits, with three
adapted conformal families. E1 represents continuous-treatment conformal
inference
\cite{SchroederEtAl2025ContinuousTreatmentCP,
VerhaegheEtAl2026DoseResponseCP}; E2 represents weighted
distribution-shift conformal inference combining treatment/intervention
and target-population reweighting without graph-dependent residual
reassignment
\cite{QinEtAl2025CovariateShiftCausal,
QiuEtAl2023CovariateShift,YangEtAl2024DoublyRobust}; and E3 represents
spatial/local conformal calibration without treatment-shift correction
\cite{MaoMartinReich2024SpatialPrediction,
Guan2023LocalizedConformal,HoreBarber2025LocalWeights}. The RF-primary M6
implementation serves as the GeoDose-CP reference.

Each comparator is evaluated under its native support conditions rather
than being forced to reproduce GeoDose-CP operating rules. In the
implemented comparison, E1 is natively nonapplicable when the queried
target dose is an exact treatment atom under its continuous-treatment
formulation; such cases are treated as nonapplicability rather than
refusal. E2 and M6 operate on the mixed treatment measure, whereas E3 is
treatment-agnostic. None of E1--E3 implements the candidate-specific joint
treatment--spatial target law of Theorem~\ref{thm:graph_local}. They are
therefore mechanism-aligned adaptations of published method families rather
than byte-for-byte reproductions of authors' software. Their exact scores,
weighting rules, localization, treatment-density specifications, support
conditions, calibration procedures, hyperparameters, and set construction
are reported in the Supplement, together with a strict all-method
common-subset comparison.

\subsection{Real NSW Earth-Observation Demonstration}
\label{subsec:real_nsw}

The NSW study evaluates real-world EO applicability rather than causal
performance of the full treatment-dependent GeoDose-CP pipeline. Three
mines---Mt Arthur Coal, Hunter Valley Operations, and Bulga Complex---were
selected at the primary 90-m support using outcome-blind EO-quality
criteria, with an independently anchored 180-m analysis for
change-of-support sensitivity. Random forest (RF) is the primary prediction
track and XGBoost the confirmation track. Across both spatial supports and
predictor tracks, the study comprises 41,244 method--query records.

The response is annual change in the observed Digital Earth Australia
photosynthetic-vegetation product,

\begin{equation}
Y_i^{\mathrm{EO}}
=
\frac{
PV_{i,2025}-PV_{i,2024}
}{100}.
\label{eq:real_outcome}
\end{equation}

This quantity represents observed satellite-product green fractional-cover
change rather than an error-free latent measure of ecological recovery
\cite{SinghEtAl2024ConformalEO,
MartinezFerrerEtAl2022BiophysicalUQ,
GarciaSoriaEtAl2024VegetationUQ}. Predictors comprise 27 pretreatment EO
and environmental variables plus three mine indicators. All 2025
predictors, outcome-derived quantities, mapped rehabilitation fraction,
oracle or hidden variables, sample-role indicators, and spatial-axis
coordinates are excluded.

A treatment-authenticity audit established that the public NSW archive
contains rehabilitation snapshots but not an auditable longitudinal
continuous rehabilitation treatment consistent with the target estimand.
No proxy exposure was introduced. Consequently, treatment-dependent
procedures M2, M4, and M6 are nonoperational. The real-data analysis
instead uses M1, M3, and a treatment-free graph-safe fallback (GS), which
calibrates on a deterministic graph-safe subset without treatment-dependent
transport. GS is distinct from the treatment-weighted MineDoseBench M5
baseline defined in Section~\ref{subsec:metrics}. The NSW study therefore
tests treatment-authenticity gating, spatial deployment, EO-product quality,
and method applicability rather than causal validation of full M6.

Geographically buffered partitioning is the primary validation design.
Random splitting is used only for the RF/M1 predictive spatial-leakage
diagnostic, leave-one-mine-out analysis as a cross-mine transfer diagnostic,
and the 90-to-180-m comparison as a change-of-support sensitivity
\cite{RobertsEtAl2017SpatialCV,
ValaviEtAl2019BlockCV,
MeyerEtAl2018TargetValidation,
PlotonEtAl2020SpatialValidation,
LudwigEtAl2023Transferability}. The 180-m substrate is reconstructed
independently under the anchored all-available-child eligibility rule
rather than by rescaling 90-m predictions, and a UE$\leq20$ analysis
provides an additional EO-product-quality sensitivity.

For M3, conformity p-values are evaluated over the complete eligible
real-data frame. Full finite-domain inversion is restricted to an
outcome-blind subset of 60 targets for 90-m RF, 30 for 90-m XGBoost,
36 for 180-m RF, and 18 for 180-m XGBoost, totaling 144 inversions.
Reported M3 widths therefore characterize this targeted inversion subset,
not full-test-frame or population-level efficiency. For the 90-m RF width comparison in Figure~\ref{fig:nsw_observed}, M1 and GS are
evaluated on this same frozen 60-target subset.

The fitted NSW spatial nuisance is empirical and not finite-sample
theorem-certified; associated N3 quantities are therefore reported only as
diagnostics, not as certified real-data coverage guarantees.

\section{Results}
\label{sec:results}

\subsection{Controlled Validation and Overall MineDoseBench Performance}
\label{subsec:results_overall}

The controlled known-truth experiments first evaluated the inferential
mechanism independently of realistic mine geometry and environmental
heterogeneity. Across 4,190 case replications and 52,740 scientific
queries, no computational failures occurred. None of the 17 prespecified
lower-tail coverage groups showed multiplicity-adjusted evidence of
undercoverage.

MineDoseBench then evaluated M2--M6 under real mine geometry and
pretreatment environmental context while retaining known counterfactual
truth. Table~\ref{tab:mdb_overall} summarizes the RF-primary results across
all 27 configurations. Full GeoDose-CP (M6) attained mean selective
coverage of 0.9692; all 27 configuration-level point estimates were at or
above the nominal 0.90 level, with a minimum of 0.9140. The 95\% replication-by-mine cluster-bootstrap lower bounds were at least
0.90 in 26/27 configurations. The sole exception was S2, for which selective coverage was
0.914 and the lower bound was 0.890.

M6 also maintained strong lower-tail local performance. Its mean local
$q_{0.05}$ was 0.9488 and its minimum was 0.8951, the highest minimum
among M2--M6. The corresponding minima were 0.8804 for M4, 0.7694 for
M5, 0.7464 for M3, and 0.7079 for M2. M3 and M4 were more conservative
on average, with mean local $q_{0.05}$ values of 0.9725 and 0.9574,
respectively, showing that strong average coverage did not guarantee
equally strong lower-tail local behavior.

M4 also achieved high empirical selective coverage, with all 27
configuration-level estimates above 0.90. This result does not establish
the factorized M4 construction as the general target reference law because,
as shown in Section~\ref{subsec:factorization}, product factorization
requires additional structural conditions. The comparison therefore
distinguishes empirical benchmark performance from the joint-law
justification supplied by G2, rather than requiring M6 to numerically
dominate every comparator.


\begin{table*}[!t]

\caption{Overall MineDoseBench Performance Across 27 Configurations}

\label{tab:mdb_overall}

\centering

\footnotesize

\setlength{\tabcolsep}{5.5pt}

\renewcommand{\arraystretch}{1.12}

\begin{tabular}{@{}lcccccc@{}}

\toprule

\textbf{Method} &
\shortstack{\textbf{Mean}\\\textbf{Sel. Cov.}} &
\shortstack{\textbf{Min.}\\\textbf{Sel. Cov.}} &
\shortstack{\textbf{Cases}\\$\boldsymbol{\geq 0.90}$} &
\shortstack{\textbf{Mean Local}\\$\boldsymbol{q_{0.05}}$} &
\shortstack{\textbf{Min. Local}\\$\boldsymbol{q_{0.05}}$} &
\shortstack{\textbf{Return}\\\textbf{Rate}} \\

\midrule

M2 & 0.8855 & 0.8022 & 8/27  & 0.8265 & 0.7079 & 0.9500 \\
M3 & 0.9856 & 0.8360 & 26/27 & 0.9725 & 0.7464 & 1.0000 \\
M4 & 0.9805 & 0.9422 & 27/27 & 0.9574 & 0.8804 & 0.9500 \\
M5 & 0.8989 & 0.8360 & 12/27 & 0.8385 & 0.7694 & 0.8945 \\

\textbf{M6} &
\textbf{0.9692} &
\textbf{0.9140} &
\textbf{27/27} &
\textbf{0.9488} &
\textbf{0.8951} &
\textbf{0.9500} \\

\bottomrule

\end{tabular}

\vspace{1mm}

\parbox{0.98\textwidth}{\footnotesize
\emph{Notes:}
Results use the RF-primary track on the observed benchmark-response scale.
Treatment-dependent primary methods use the frozen estimated
mixed-propensity model $\widehat g(A\mid U)$; the oracle treatment law is
not the primary MineDoseBench track. Sel.\ Cov.\ denotes selective coverage
conditional on a returned prediction set. Local $q_{0.05}$ is the fifth
percentile of target-local selective coverage. Means are calculated across
the 27 configurations. M4 is the naive-product comparator and is not
generally theorem-backed. Width is not reported because efficiency is
interpreted only under matched selective coverage.
}
\end{table*}

Selective coverage is interpreted jointly with return rate because refused
queries are reported separately rather than counted as noncoverage.

Figure~\ref{fig:configurations} resolves the aggregate results by
configuration. Panel (a) shows M6 selective coverage across all 27
configurations, panel (b) compares local $q_{0.05}$ across M2--M6, and
panel (c) reports the corresponding refusal rates. M6 retains the strongest
minimum local $q_{0.05}$, while refusal remains limited across most
supported configurations and increases as information or support
deteriorates.


\begin{figure*}[!t]

\centering

\includegraphics[width=0.98\textwidth]{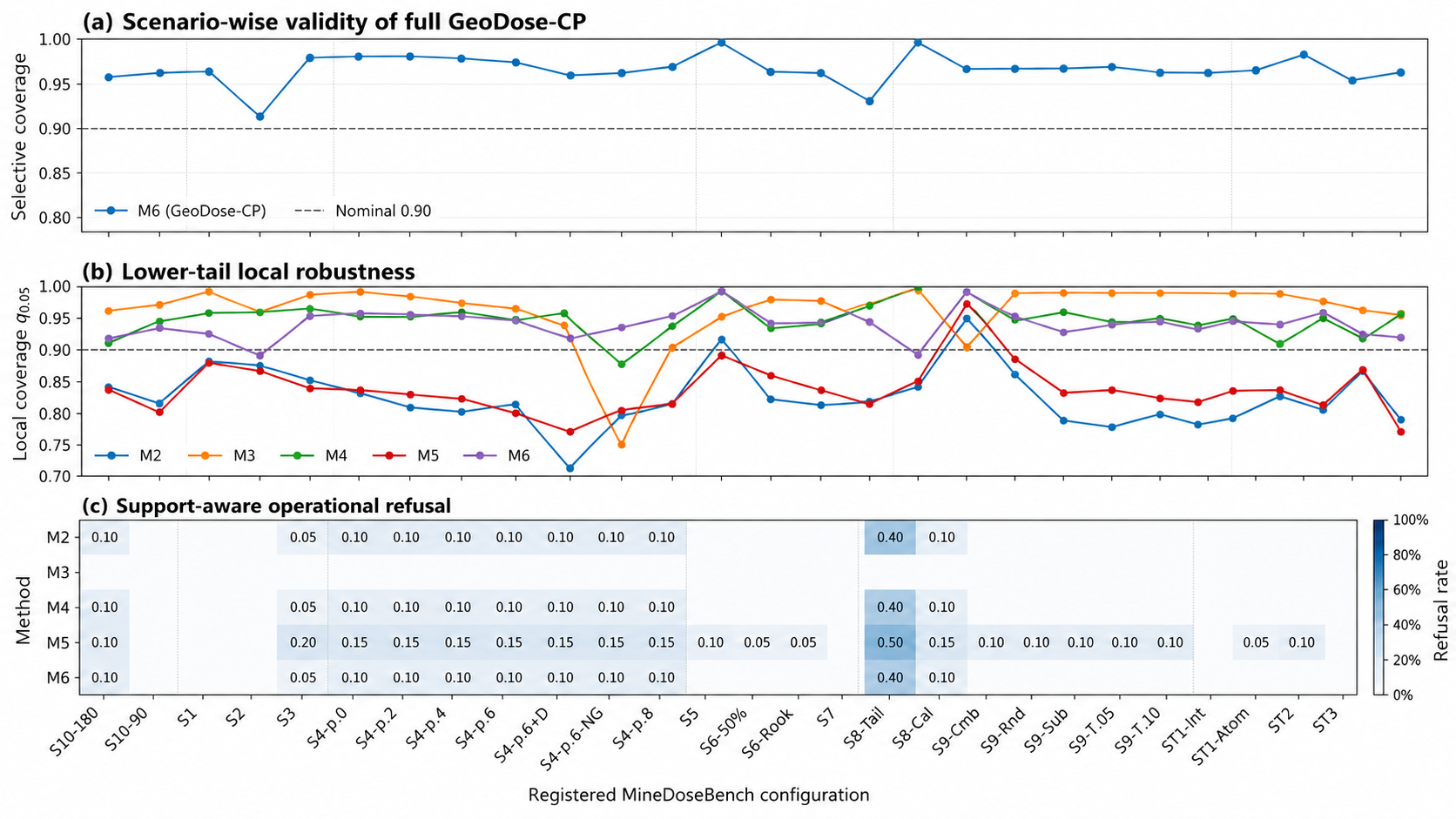}

\caption{MineDoseBench selective coverage, local robustness, and
support-aware refusal across 27 configurations.
(a) Configuration-wise selective coverage of full GeoDose-CP (M6).
(b) Fifth-percentile target-local selective coverage for M2--M6.
(c) Operational refusal rates for M2--M6. Dashed horizontal lines denote
nominal 0.90 coverage. Results correspond to the RF-primary evaluation on
the observed benchmark-response scale.}

\label{fig:configurations}

\end{figure*}

\subsection{Spatial Dependence and Coverage-Matched Efficiency}
\label{subsec:results_spatial}

Across the S4 dependence sequence, M6 selective coverage remained above
the nominal 0.90 level throughout, taking values 0.9844, 0.9844, 0.9822,
0.9778, and 0.9733 at $\rho=0$, 0.2, 0.4, 0.6, and 0.8, respectively
(Fig.~\ref{fig:s4_hero}). Thus, empirical coverage declined only modestly
as spatial dependence strengthened.

Efficiency showed a different pattern. Under the coverage-matching rule in
\eqref{eq:matched_coverage}, M6 was 3.8\%, 4.9\%, and 11.4\% narrower than
M5 at $\rho=0$, 0.2, and 0.4, respectively, but 31.2\% wider at
$\rho=0.6$. At $\rho=0.8$, the procedures did not satisfy the 0.03
coverage-matching tolerance, so no efficiency comparison is made.

These width comparisons use the same 50 prespecified S4 target identities,
with 10 identities at each
$\rho\in\{0,0.2,0.4,0.6,0.8\}$. Full finite-domain inversion across M3,
M4, and M6 yields 150 inversion records; M2 and M5 contribute their
closed-form returned intervals on the same registered identities,
intersected with the benchmark response domain
$\mathcal{Y}_{\mathrm{MDB}}=[-1,1]$, for the efficiency-only comparison.
The resulting widths therefore provide targeted structural-efficiency
evidence rather than population-level sharpness estimates. Overall, M6
maintained stable selective coverage as spatial dependence increased,
while its efficiency relative to M5 remained regime-dependent.

\begin{figure*}[!t]
\centering
\includegraphics[width=\textwidth]{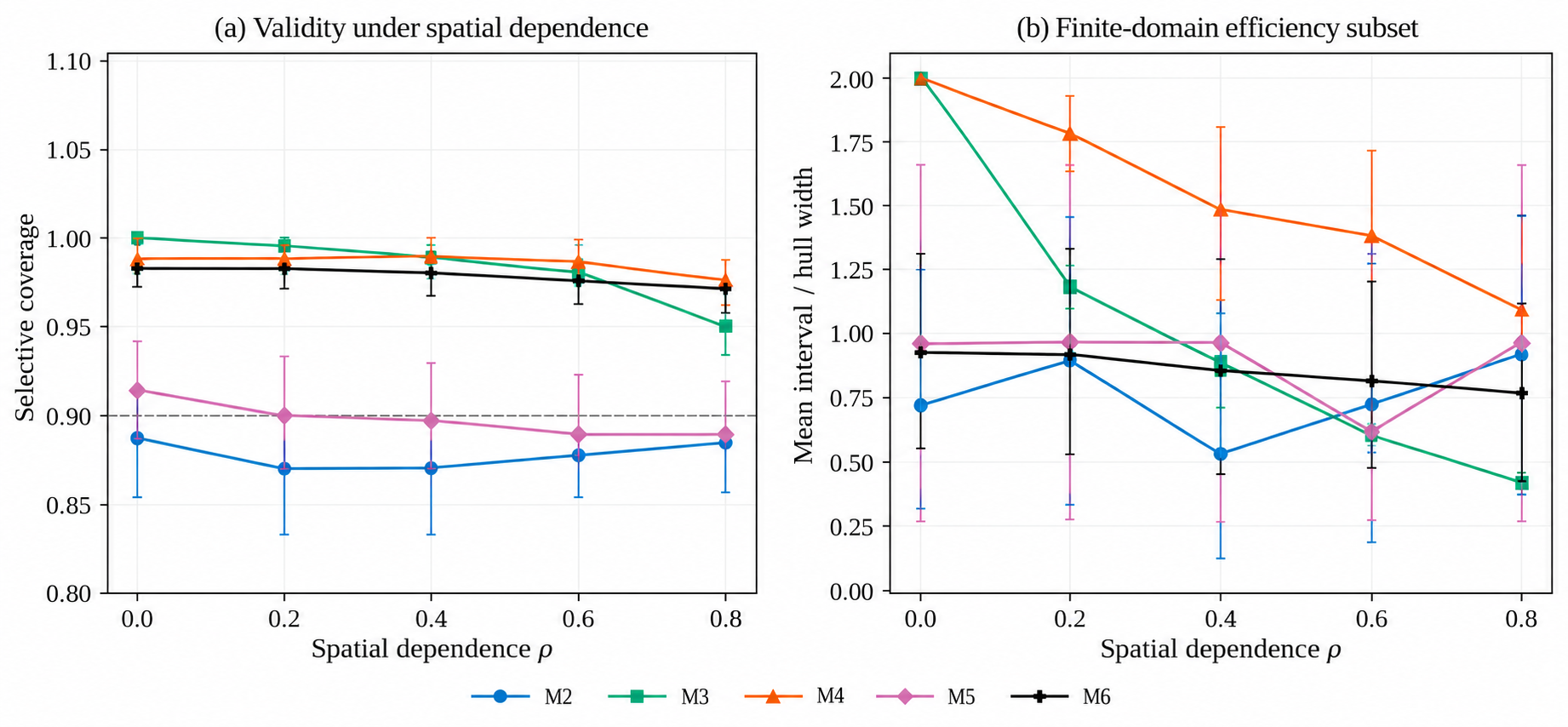}

\caption{Effect of increasing spatial dependence in S4.
(a) RF-primary selective coverage of M2--M6 for
$\rho\in\{0,0.2,0.4,0.6,0.8\}$; error bars denote 95\%
replication-by-mine cluster-bootstrap confidence intervals.
(b) Mean finite-domain interval/hull width on
$\mathcal{Y}_{\mathrm{MDB}}=[-1,1]$ for the same prespecified S4
target identities. For M3, M4, and M6, widths are obtained from full
finite-domain candidate inversion; for M2 and M5, closed-form returned
intervals are intersected with $\mathcal{Y}_{\mathrm{MDB}}$ for this
efficiency-only representation. Error bars denote 95\%
replication-by-mine cluster-bootstrap confidence intervals for mean width.
Width is interpreted as efficiency only when the selective-coverage
difference satisfies the prespecified 0.03 coverage-matching criterion;
the $\rho=0.8$ M6--M5 comparison is therefore not interpreted as an
efficiency comparison.}

\label{fig:s4_hero}
\end{figure*}
\subsection{Approximation, Stress, and Support Boundaries}
\label{subsec:results_stress}

Exact and sparse M6 showed close aggregate and decision-level agreement,
although pointwise equality was not claimed. Across 2,700 paired target
identities, selective coverage was 0.9885 for exact inference and 0.9874
for the sparse $m=64$ route. Mean absolute p-value discrepancy was 0.00557
and the median was numerically zero, with

\begin{equation}
q_{0.90}=0.01607,\qquad
q_{0.95}=0.03467,\qquad
q_{0.99}=0.09739,
\label{eq:exact_sparse_quantiles}
\end{equation}

and maximum $|\Delta p|=0.42945$. Discrepancies exceeded 0.01, 0.05,
and 0.10 for 352 (13.0\%), 76 (2.81\%), and 23 (0.85\%) targets,
respectively, yet only 9/2,700 (0.33\%) changed inclusion status at
$\alpha=0.10$. Thus, sparse inference remained close to the exact route at
the aggregate and decision levels despite a small number of materially
larger pointwise discrepancies. The largest occurred in S2
($p_{\mathrm{exact}}=0.8702$, $p_{\mathrm{sparse}}=0.4408$) without
changing the inclusion decision.

Within S4, approximation sensitivity increased with spatial dependence: mean $|\Delta p|$ was
approximately zero at $\rho=0$ and 0.00159, 0.00318, 0.00451, and
0.00690 at $\rho=0.2$, 0.4, 0.6, and 0.8, respectively; two decision
disagreements occurred at $\rho=0.8$. Unambiguous target-level linkage
between discrepancy and ESS/maximum-weight diagnostics was available for
only 51/2,700 identities, so no general diagnostic association is inferred.
The $m=64$ route therefore provides close aggregate approximation while
retaining nonzero query-level approximation error.

\begin{figure}[!t]
\centering
\includegraphics[width=\columnwidth]{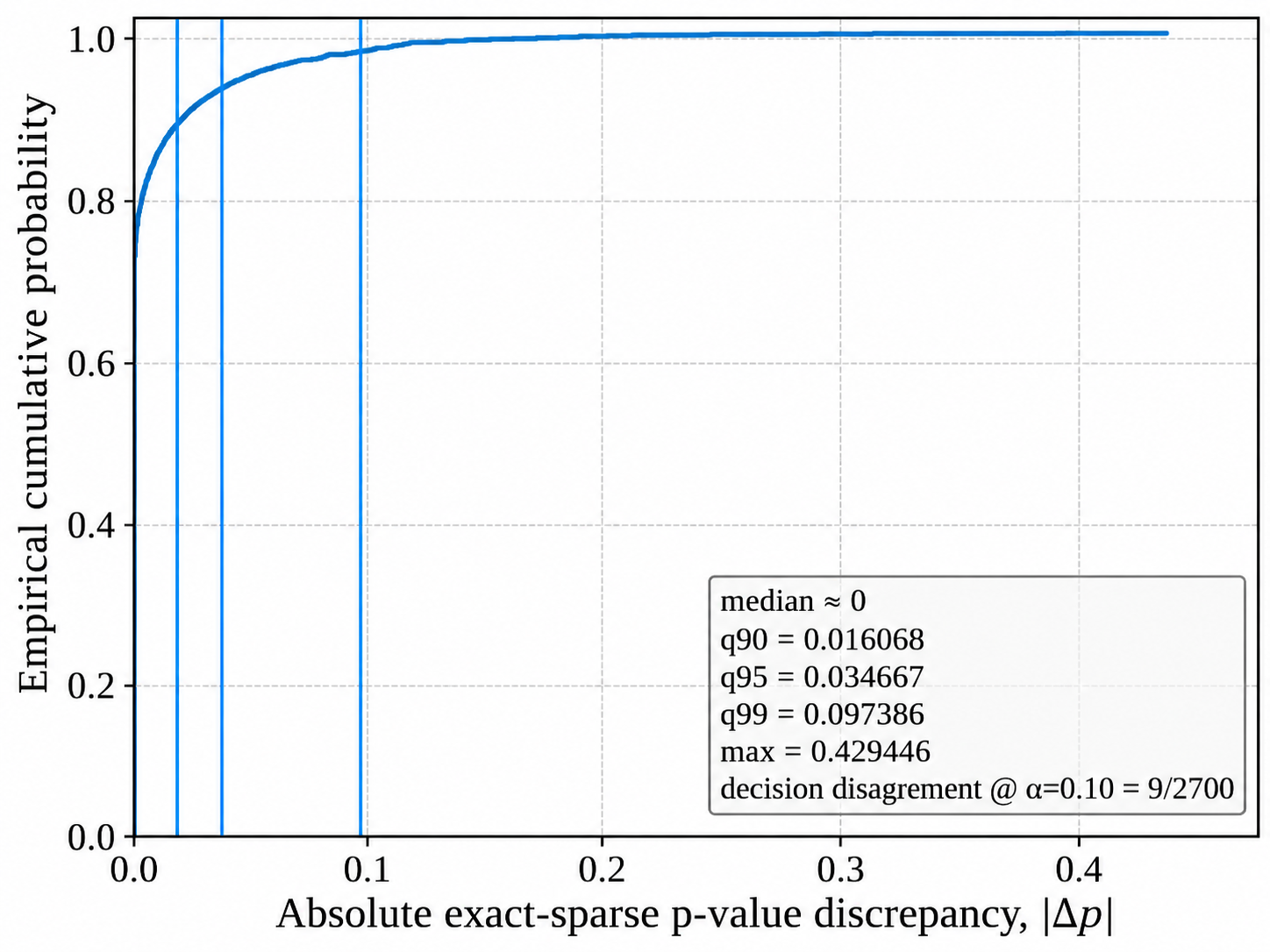}
\caption{Exact--sparse M6 discrepancy across 2,700 paired registered
target identities. The empirical cumulative distribution function shows
the absolute difference $|\Delta p|$ between full-graph exact inference
and the sparse $m=64$ route. Vertical reference lines mark the
0.90, 0.95, and 0.99 quantiles
($0.01607$, $0.03467$, and $0.09739$, respectively).
The median discrepancy is numerically zero and the maximum is $0.42945$.
Only 9 of 2,700 targets (0.33\%) changed inclusion status at
$\alpha=0.10$, so close aggregate and decision-level agreement does not
imply pointwise numerical equivalence.}
\label{fig:exact_sparse_ecdf}
\end{figure}

A complementary structural audit examined the product construction in
\eqref{eq:naive_product}. Under the reduction condition, the factorized
and joint graph-local constructions agreed to numerical precision; under
the nonfactorizing treatment--spatial regime, their candidate-specific
relationship became nonconstant. The audit therefore recovers
factorization in its valid special case while confirming that it does not,
in general, reproduce the G2 target-orbit law.

Table~\ref{tab:mdb_stress_audits} summarizes the principal robustness and
support-boundary results. M6 remained above nominal selective coverage
under target-design shift (0.9644), non-Gaussian dependence (0.9667), five
measurement-error perturbations (0.966--0.972), and the 90-m/180-m
change-of-support analysis (0.9640/0.9578), with return rates of 100\% and
90\% at the two spatial supports, respectively.

Support limitations produced distinct operating behavior. Genuine endpoint
atoms at $A=0$ and $A=1$ yielded selective coverage of 0.980 and 0.990
with 100\% return, whereas the corresponding interior-only law refused
every endpoint query rather than borrowing support from nearby continuous
mass. Under severe-tail stress, M6 returned 60\% of queries with selective
coverage 1.000, while false support remained 0.05. Leave-one-mine-out
transfer produced 100\% refusal across all 375 M6 queries because the
held-out mine lacked a connected local calibration orbit. Hidden spatial
confounding remained a negative control: predictive selective coverage was
0.966, but this does not restore the conditional ignorability required for
causal interpretation.


\begin{table*}[!t]

\caption{Stress, Support, and Applicability Audits for GeoDose-CP}
\label{tab:mdb_stress_audits}

\centering
\footnotesize
\renewcommand{\arraystretch}{1.10}
\setlength{\tabcolsep}{4pt}

\begin{tabularx}{\textwidth}{
@{}
>{\raggedright\arraybackslash}p{1.42in}
>{\raggedright\arraybackslash}p{1.42in}
>{\raggedright\arraybackslash}p{1.52in}
>{\raggedright\arraybackslash}X
@{}}

\toprule

\textbf{Audit} &
\textbf{Selective coverage} &
\textbf{Return/support} &
\textbf{Principal finding} \\

\midrule

Target-design shift &
0.9644 &
90\% return &
Coverage remains above the nominal 0.90 level under nonidentity
target-design shift. \\

Non-Gaussian spatial law &
0.9667 &
90\% return &
Coverage remains above nominal under the non-Gaussian residual-dependence
setting. \\

Measurement-error perturbations &
0.966--0.972 &
100\% return &
Coverage remains above nominal across all five measurement-error
perturbations. \\

Change of spatial support &
90 m: 0.9640;\newline
180 m: 0.9578 &
90 m: 100\% return;\newline
180 m: 90\% return &
Coverage remains above nominal at both independently anchored spatial
supports, with lower return at 180 m. \\

Supported endpoint atoms &
$A=0$: 0.980;\newline
$A=1$: 0.990 &
100\% return at both atoms &
Both genuine endpoint atoms remain supported under the mixed treatment
measure. \\

Unsupported endpoints &
--- &
100\% refusal at $A=0,1$ &
The interior-only treatment law refuses unsupported endpoint queries
rather than assigning artificial endpoint mass. \\

Severe-tail support stress &
1.000 &
60\% return;\newline
false support $=0.05$ &
Selective coverage is retained among returned queries, but support loss
reduces return and residual false support remains. \\

Hidden spatial confounding &
0.966 &
100\% return;\newline
false support $=0.05$ &
Negative control: predictive coverage remains high despite violation of
the causal-identification condition. \\

Leave-one-mine-out &
--- &
0\% return;\newline
375/375 refused &
All held-out-mine M6 queries refuse because a connected local calibration
orbit is unavailable; successful cross-mine local-orbit transfer is not
demonstrated. \\

\bottomrule

\end{tabularx}

\vspace{1mm}

\parbox{0.98\textwidth}{\footnotesize
\emph{Notes:}
Coverage denotes selective coverage among returned prediction sets;
refused queries are reported separately and are not counted as
noncoverage. ``---'' indicates that selective coverage is undefined
because no prediction sets were returned. False support denotes the
prespecified rate at which a benchmark query with a known support failure
nevertheless passes the applicable operational support gate. The
hidden-confounding experiment is a negative control and does not provide
evidence of causal identification.
}

\end{table*}

\subsection{Treatment-Density Sensitivity and External Comparators}
\label{subsec:results_sensitivity_external}

\subsubsection{Treatment-Density Sensitivity}

Holding all non-treatment components of the RF-primary M6 pipeline fixed,
the estimated treatment law closely tracked the oracle over the selected
seven-configuration subset (Table~\ref{tab:g_sensitivity}). Mean selective
coverage was 0.9830 for $G_{\mathrm{ORACLE}}$ and 0.9825 for
$G_{\mathrm{ESTIMATED}}$, with mean local $q_{0.05}$ of 0.9661 and
0.9655, respectively. Their minimum selective coverages were 0.9680 and
0.9644, and minimum local $q_{0.05}$ values were 0.9347 and 0.9267.

Misspecification was more visible locally than globally.
$G_{\mathrm{MISSPECIFIED}}$ retained mean selective coverage of 0.9786,
but its mean and minimum local $q_{0.05}$ decreased to 0.9545 and 0.8882.
Under strong dependence combined with target-design shift, local
$q_{0.05}$ was 0.9347, 0.9267, and 0.8882 for the oracle, estimated,
and misspecified laws, respectively.

\begin{table*}[!t]
\caption{Treatment-Density Sensitivity for M6}
\label{tab:g_sensitivity}
\centering
\footnotesize
\renewcommand{\arraystretch}{1.10}
\setlength{\tabcolsep}{7pt}

\begin{tabular}{lccccc}
\toprule
\textbf{Treatment law} &
\textbf{Mean sel. cov.} &
\textbf{Min. sel. cov.} &
\textbf{Mean local $q_{0.05}$} &
\textbf{Min. local $q_{0.05}$} &
\textbf{Return rate} \\
\midrule
$G_{\mathrm{ORACLE}}$
    & 0.9830 & 0.9680 & 0.9661 & 0.9347 & 0.9571 \\
$G_{\mathrm{ESTIMATED}}$
    & 0.9825 & 0.9644 & 0.9655 & 0.9267 & 0.9071 \\
$G_{\mathrm{MISSPECIFIED}}$
    & 0.9786 & 0.9537 & 0.9545 & 0.8882 & 0.9857 \\
\bottomrule
\end{tabular}

\vspace{1mm}

\parbox{0.96\textwidth}{\footnotesize
\emph{Notes:}
Results use the same seven configurations and 20 replications for each
treatment-law specification, with all non-treatment components held fixed.
$G_{\mathrm{ESTIMATED}}$ is the treatment model used by the registered
RF-primary M6 pipeline. Coverage is selective among returned prediction
sets.
}
\end{table*}

Support diagnostics did not identify misspecification by themselves.
Mean support ESS was 146.46, 89.17, and 415.47 for
$G_{\mathrm{ORACLE}}$, $G_{\mathrm{ESTIMATED}}$, and
$G_{\mathrm{MISSPECIFIED}}$, respectively, while mean maximum normalized
weight was 0.1513, 0.2203, and 0.0714. The misspecified law therefore
appeared most favorable by weight dispersion despite producing the weakest
lower-tail local coverage. Large ESS and diffuse normalized weights indicate
favorable support concentration, but they do not establish correctness of
the treatment model $g(A\mid U)$.

\subsubsection{External Comparator Evaluation}

Table~\ref{tab:external_baselines} summarizes M6 and three adapted
conformal families on the selected MineDoseBench external-comparator
subset: continuous-treatment inference (E1), weighted distribution-shift
inference (E2), and spatial/local calibration (E3). Because E1 is natively
inapplicable to an exact atomic target query, its aggregate statistics use
six applicable configurations, whereas E2, E3, and M6 use all seven.
A strict six-configuration all-method comparison is reported in the
Supplement. On this strict common subset, mean selective coverage remained ordered
M6 (0.9796) $>$ E3 (0.9510) $>$ E2 (0.9153) $>$ E1 (0.8970),
and the same ordering held for mean local $q_{0.05}$.

On this descriptive selected-subset summary, mean selective coverage was
0.8970, 0.9191, 0.9463, and 0.9825 for E1, E2, E3, and M6,
respectively; the corresponding minimum values were 0.8120, 0.8580,
0.9180, and 0.9644. Mean local $q_{0.05}$ was 0.8310, 0.8701,
0.9134, and 0.9655, with minima of 0.7193, 0.7940, 0.8817, and
0.9267, respectively. E3 was the strongest external comparator, while M6
showed the strongest selective-coverage and lower-tail local profile on the
evaluated supported regimes.

\begin{table*}[!t]
\caption{External-Comparator Performance on the Selected MineDoseBench Evaluation Subset}
\label{tab:external_baselines}
\centering
\footnotesize
\renewcommand{\arraystretch}{1.10}
\setlength{\tabcolsep}{4pt}

\begin{tabular}{lcccccccc}
\toprule
\textbf{Method} &
\textbf{Mean} &
\textbf{Min.} &
\textbf{Cases} &
\textbf{Mean local} &
\textbf{Min. local} &
\textbf{Native} &
\textbf{Return} &
\textbf{Infinite-set} \\
&
\textbf{sel. cov.} &
\textbf{sel. cov.} &
$\boldsymbol{\geq0.90}$ &
$\boldsymbol{q_{0.05}}$ &
$\boldsymbol{q_{0.05}}$ &
\textbf{applicability} &
$\boldsymbol{\mid}$ \textbf{applicable} &
$\boldsymbol{\mid}$ \textbf{returned} \\
\midrule
E1 continuous-treatment
    & 0.8970 & 0.8120 & 3/6
    & 0.8310 & 0.7193
    & 0.8571 & 1.0000 & 0.3097 \\

E2 weighted shift
    & 0.9191 & 0.8580 & 5/7
    & 0.8701 & 0.7940
    & 1.0000 & 1.0000 & 0.2714 \\

E3 spatial/local
    & 0.9463 & 0.9180 & 7/7
    & 0.9134 & 0.8817
    & 1.0000 & 1.0000 & 0.0000 \\

M6 GeoDose-CP
    & \textbf{0.9825} & \textbf{0.9644} & \textbf{7/7}
    & \textbf{0.9655} & \textbf{0.9267}
    & 1.0000 & 0.9071 & 0.0000 \\
\bottomrule
\end{tabular}

\vspace{1mm}

\parbox{0.97\textwidth}{\footnotesize
\emph{Notes:}
E1--E3 are mechanism-aligned adaptations of published method families,
not exact software reproductions. Native applicability is the proportion
of registered queries for which a method is defined under its own support
conditions. Return rate is calculated conditional on native applicability.
E1 is natively inapplicable to the exact atomic-target configuration under
its implemented continuous-treatment formulation; this is nonapplicability,
not refusal. Consequently, E1 contributes six configurations to its
coverage summaries and has native applicability 0.8571 but return rate
1.0000 among applicable queries. Infinite-set rate is calculated among
returned prediction sets and denotes vacuous or unbounded returned sets.
M6 uses $G_{\mathrm{ESTIMATED}}$. A strict six-configuration all-method
comparison is provided in the Supplement.
}
\end{table*}
The joint treatment--spatial stresses produced the clearest separation.
At $\rho=0.6$, selective coverage was 0.872, 0.858, 0.958, and 0.9778
for E1, E2, E3, and M6. With target-design shift added, coverage was
0.812, 0.934, 0.950, and 0.9644, while local $q_{0.05}$ was 0.7193,
0.8920, 0.9253, and 0.9267, respectively. These regimes are especially
informative because treatment shift and spatial dependence act jointly,
the setting for which G2 constructs a candidate-specific joint reference
law rather than combining separately derived treatment and spatial
weights.

Operational behavior also differed (Fig.~\ref{fig:external_baselines}).
Among queries returned by both M6 and E3, M6 achieved higher selective
coverage in all seven configurations, with casewise advantages of
approximately 1.0--8.2 percentage points. Under severe-tail stress, E1
and E2 produced infinite sets for more than 93\% of applicable queries,
whereas M6 returned finite sets for 60\% and refused the remainder, with
selective coverage 1.000 among returned sets. Thus, refusal and vacuous
return represent distinct responses to deteriorating support.

\begin{figure*}[!t]
\centering
\includegraphics[width=0.96\textwidth]
{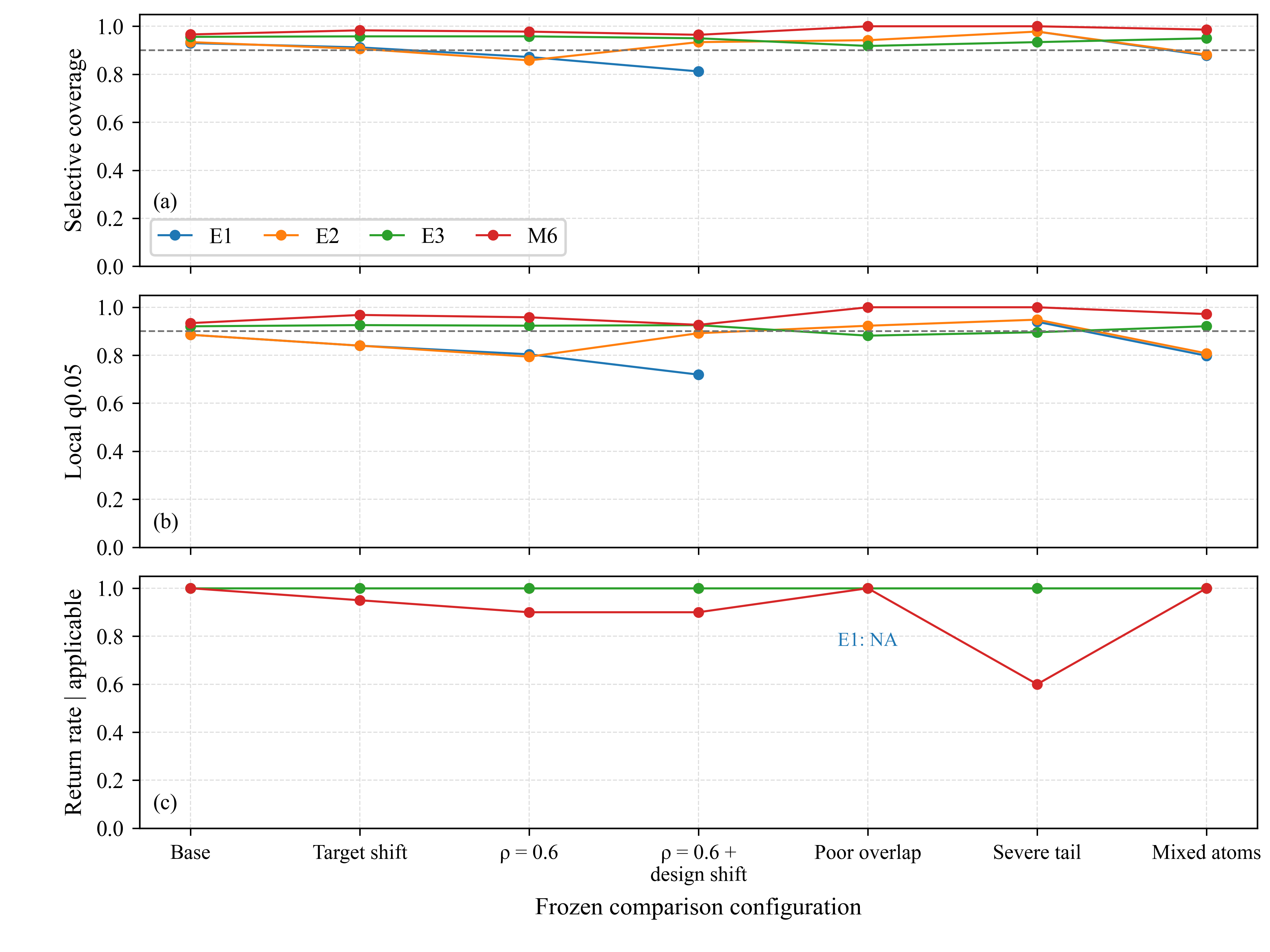}
\caption{External-comparator performance on the selected MineDoseBench
evaluation subset. The three panels show configuration-level selective
coverage, lower-tail local $q_{0.05}$, and return rate conditional on
native applicability, respectively. E1, E2, and E3 denote the adapted
continuous-treatment, weighted distribution-shift, and spatial/local
conformal families; M6 denotes GeoDose-CP using the primary estimated
treatment density. Dashed horizontal lines in the upper two panels denote
nominal 0.90 coverage. E1 is natively inapplicable in the poor-overlap
atomic-target configuration, so no conditional return-rate value is
defined there; this is nonapplicability rather than refusal.}
\label{fig:external_baselines}
\end{figure*}

Prediction-set width is interpreted only when
\eqref{eq:matched_coverage} is satisfied. On the restricted
coverage-matched inversion subset, E3 can produce substantially sharper
sets than M6. The external comparison therefore supports a stronger
selective-coverage and lower-tail local robustness profile for M6 in the
evaluated supported regimes, while providing no basis for a claim of
universal efficiency superiority.
\subsection{Real NSW Observed-Product Results}
\label{subsec:results_nsw}

The NSW archive did not provide the longitudinal continuous rehabilitation
treatment required by the target estimand. Consequently, treatment-dependent
procedures M2, M4, and M6 were nonoperational and no proxy treatment was
introduced. The real-data results therefore concern M1, M3, and the
treatment-free graph-safe fallback (GS) as predictive or spatial-calibration
procedures for the observed EO product. GS is distinct from the
treatment-weighted MineDoseBench M5 baseline.

Figure~\ref{fig:nsw_spatial} shows the geographically buffered 90-m design
for Mt Arthur Coal, Hunter Valley Operations (HVO), and Bulga Complex,
together with M3 conformity p-values for held-out RF-primary targets.
These p-values characterize conformity with the fitted spatial reference
and are interpreted as empirical diagnostics rather than theorem-certified
coverage guarantees or causal evidence.


\begin{figure*}[!t]
\centering
\includegraphics[width=\textwidth]{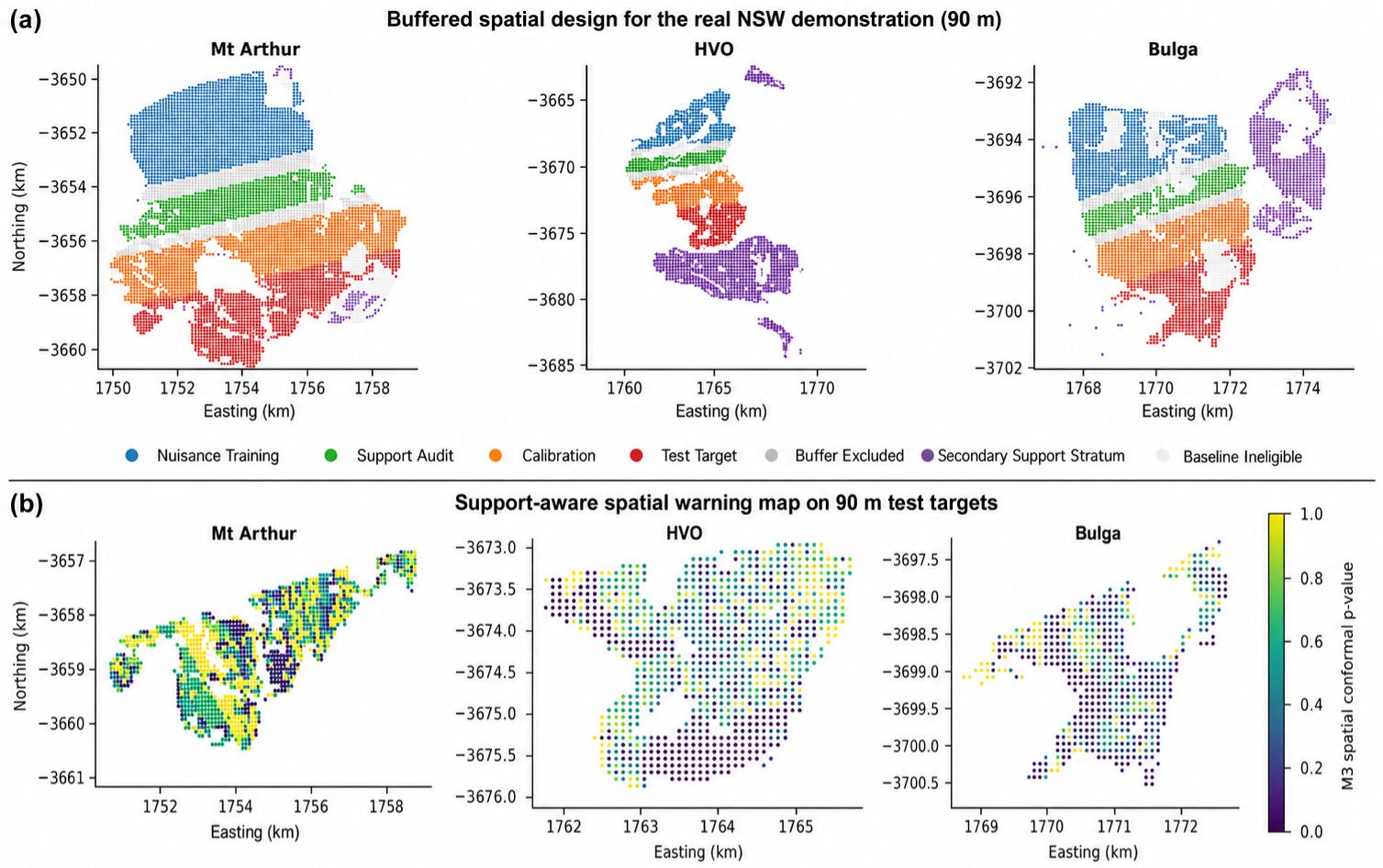}
\caption{Spatial design and empirical uncertainty diagnostics for the NSW
demonstration at 90-m support.
(a) Buffered partition for Mt Arthur Coal, HVO, and Bulga Complex,
including nuisance-training, support, calibration, held-out test, buffer,
secondary, and ineligible regions.
(b) M3 conformity p-values for held-out RF-primary targets; lower values
indicate weaker conformity with the fitted spatial reference. The p-values
are empirical diagnostics, not causal evidence or theorem-certified
coverage guarantees.}
\label{fig:nsw_spatial}
\end{figure*}

Observed-product coverage varied substantially with spatial support and
predictor. Under RF, pooled M1/M3/GS coverage was
0.8381/0.8154/0.8447 at 90 m and
0.9401/0.8231/0.9201 at 180 m. Under XGBoost, the corresponding values
were 0.9028/0.9050/0.8644 and 0.9472/0.8545/0.9344. These changes
indicate support- and model-sensitivity rather than superiority of the
coarser resolution.

Geographic heterogeneity was also pronounced. In the 90-m RF analysis,
M1/M3/GS coverage was 0.9595/0.9285/0.9619 at Mt Arthur,
0.8028/0.7306/0.7780 at HVO, and 0.6193/0.6995/0.6885 at Bulga.
Model sensitivity was site dependent: under 90-m XGBoost, Bulga
M1/M3 coverage increased to 0.9927/0.9873, whereas HVO remained at
0.7963/0.7909. The real-data study therefore does not support a universal
0.90 observed-product coverage claim.

The fitted spatial nuisance also indicated a demanding dependence regime.
Ten of 12 mine$\times$scale$\times$predictor fits selected
$\rho=0.8$, the upper limit of the specified grid; the exceptions were
Bulga at 180 m, with $\rho=0.4$ for RF and $\rho=0.6$ for XGBoost.
These fitted quantities remain empirical rather than finite-sample
theorem-certified and indicate strong residual dependence and/or spatial
model mismatch.

The stricter UE$\leq20$ sensitivity retained 94.6--98.3\% of eligible
90-m cells and 88.7--96.5\% at 180 m across the three mines, with only
small qualitative changes to the main coverage patterns. The principal
site, support, and model sensitivities were therefore not attributable
solely to the primary UE eligibility threshold.

Figure~\ref{fig:nsw_observed} summarizes observed-product coverage and width for the applicable
procedures under the primary 90-m RF analysis. Coverage in panel (a) is
summarized over the full held-out test frame. For panel (b), M1, M3, and
GS widths are evaluated on the same frozen 60-target outcome-blind subset
used for finite-domain M3 inversion, yielding mean widths of 0.3097,
0.3575, and 0.3247, respectively. These widths therefore characterize a
common targeted subset rather than full-test-frame or population-level
efficiency.


\begin{figure*}[!t]
\centering
\includegraphics[width=\textwidth]{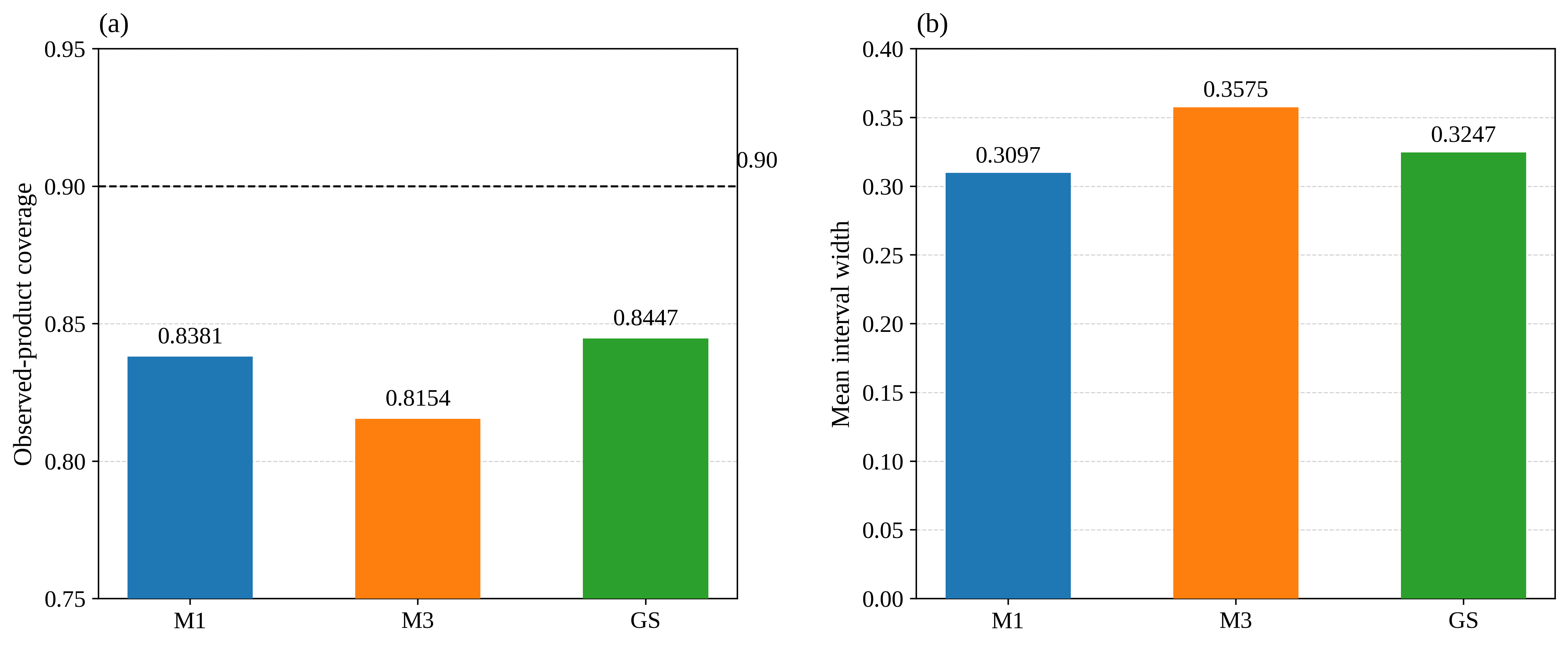}
\caption{Held-out NSW observed-product uncertainty at 90-m support under
the primary RF track. (a) Pooled observed-product coverage for M1, M3,
and the treatment-free graph-safe fallback GS over the held-out test
frame; the dashed horizontal line denotes nominal 0.90 coverage.
(b) Mean interval width for M1, M3, and GS evaluated on the same frozen
60-target outcome-blind subset used for finite-domain M3 inversion.
The width comparison is therefore restricted to this common target subset
and is not interpreted as population-level efficiency. M2, M4, and M6
are nonoperational because an authentic longitudinal rehabilitation
treatment is unavailable in the NSW archive.}
\label{fig:nsw_observed}
\end{figure*}

Validation geometry materially altered apparent predictive performance.
In the RF/M1 predictive leakage diagnostic, random splitting gave coverage
0.8957, mean width 0.1279, and RMSE 0.0540, compared with 0.8381,
0.3097, and 0.1016 under geographically buffered validation. The
substantially narrower intervals and lower error under random splitting
are consistent with optimistic assessment when training and test
observations remain spatially proximate
\cite{RobertsEtAl2017SpatialCV,
ValaviEtAl2019BlockCV,
PlotonEtAl2020SpatialValidation}.

Cross-mine evaluation provided a further applicability boundary: the
held-out mine did not support a connected local calibration orbit for the
graph-dependent route, so no successful cross-mine local-orbit transfer
was established. Together, the NSW results show substantial sensitivity
to site, predictor, spatial support, and validation geometry, while also
demonstrating that abundant EO observations do not substitute for the
treatment provenance required for causal intervention inference.
\section{Discussion}
\label{sec:discussion}

\subsection{What GeoDose-CP Establishes}
\label{subsec:discussion_establishes}

The main contribution of GeoDose-CP is structural. When continuous-treatment
shift and spatial dependence act jointly, the conformal target reassignment
law does not generally reduce to independently constructed treatment and
spatial weights. G2 instead places the intervention/observational treatment
likelihood, inverse outcome-scale Jacobian, and graph-dependent residual
likelihood in a single candidate-specific reference law. Exact candidate
inversion (G3), sparse approximation (N2), practical coverage transfer (N3),
and support-aware refusal are downstream layers built on that law rather
than substitutes for it.

The reduction analysis also defines the scope of this claim. Under
additional conditions compatible with factorization, the joint construction
reduces to the corresponding product form; outside that regime, the
candidate-specific relationship is nonconstant. GeoDose-CP therefore does
not claim that product weighting is always invalid, but that separately
constructed treatment and spatial weights are insufficient to justify the
general joint target law.

The empirical evidence is consistent with this distinction. M6 maintained
nominal-or-higher configuration-level selective coverage across all 27
MineDoseBench configurations and achieved the strongest minimum local
$q_{0.05}$ among M2--M6. On the selected external-comparator subset, M6
also showed the strongest minimum selective coverage and lower-tail local
robustness, while E3 was the strongest external comparator. M3 and M4
exhibited higher average local coverage, and M4 performed well empirically
despite lacking general G2 justification. Accordingly, the evidence
supports a strong supported-regime selective-coverage and lower-tail
local-robustness profile for the joint construction, rather than universal
numerical dominance.

The strong performance of M4 is therefore informative rather than
contradictory. A factorized method can perform well on a finite benchmark
without representing the reference law implied by the joint
treatment--spatial model. The distinction between M4 and M6 is ultimately
one of structural justification: M4 is an empirical comparator, whereas G2
specifies the candidate-specific target law under the stated joint model.

Causal identification remains a separate requirement. In the hidden
spatial-confounding negative control, predictive coverage remained high
despite violation of supported ignorability. Conformal calibration can
therefore quantify uncertainty for a predictive law even when that law
does not identify the intended causal intervention. GeoDose-CP accordingly
separates identification of the supported stochastic potential outcome from
validity of its conformal prediction set
\cite{LeiCandes2021Counterfactuals,
QinEtAl2025CovariateShiftCausal}.

\subsection{Validity, Refusal, and Efficiency}
\label{subsec:discussion_validity}

Coverage, refusal, and efficiency represent different inferential
properties and should not be collapsed into a single ranking. GeoDose-CP
often prioritizes validity and support fidelity over universal sharpness,
so strong selective coverage may coexist with wider sets or reduced return
when the query approaches the limits of the supported target law.

The endpoint and severe-tail experiments illustrate this behavior. Genuine
atoms at $A=0$ and $A=1$ were inferentially supported, whereas the same
endpoint queries were refused under an interior-only treatment law rather
than approximated from nearby continuous observations. Under severe-tail
stress, M6 returned finite sets for 60\% of queries and attained selective
coverage of 1.000 among those returned. By contrast, E1 and E2 produced
infinite sets for more than 93\% of their applicable severe-tail queries.
Explicit refusal and vacuous inference can therefore preserve coverage in
fundamentally different ways.

Refusal is not itself a guarantee of correct support assessment. A
false-support rate of 0.05 remained under severe-tail stress, and observable
support diagnostics cannot certify causal assumptions such as conditional
ignorability. ESS, weight concentration, graph support, and endpoint checks
are therefore operational diagnostics; they identify important observable
failure modes but cannot replace assumptions about the treatment and
outcome-generating process.

Efficiency is likewise regime-dependent. Under the matched-coverage rule
in \eqref{eq:matched_coverage}, M6 was 3.8\%, 4.9\%, and 11.4\% narrower
than M5 at $\rho=0$, 0.2, and 0.4, respectively, but 31.2\% wider at
$\rho=0.6$; at $\rho=0.8$, the methods were not coverage-matched and no
efficiency comparison was warranted. These results arise from the targeted
finite-domain inversion audit and should not be interpreted as
population-level sharpness estimates.

The external comparison reinforces the same boundary. On coverage-matched
restricted inversions, the spatial/local comparator E3 can be substantially
sharper than M6. The strongest empirical case for GeoDose-CP is therefore
the combination of supported-regime selective coverage, lower-tail local
robustness, and explicit behavior when support deteriorates---not universal
interval efficiency.
\subsection{Scalable Inference and Practical Certification}
\label{subsec:discussion_scalable}

The exact--sparse audit supports scalable GeoDose-CP at the aggregate
level but not pointwise equivalence. Exact and sparse M6 selective coverage
were nearly identical (0.9885 and 0.9874), and only 9/2,700 targets
(0.33\%) changed inclusion status at $\alpha=0.10$. Nevertheless,
individual discrepancies could be much larger, with
$\max|\Delta p|=0.42945$. The sparse route should therefore be interpreted
as an accurate aggregate approximation whose query-level error remains
nonzero.

Approximation sensitivity increased with spatial dependence across S4,
consistent with a harder sparse representation as residual dependence
strengthened. This pattern should not be attributed generally to low ESS
or concentrated weights: target-level linkage to those diagnostics was
unambiguous for only 51/2,700 exact-audit identities. The present evidence
therefore supports a dependence-related approximation pattern without
establishing a universal mechanism for the largest pointwise errors.

Treatment-density estimation defines a separate practical uncertainty
layer. $G_{\mathrm{ESTIMATED}}$ closely tracked
$G_{\mathrm{ORACLE}}$ in aggregate, whereas deliberate misspecification
weakened lower-tail local performance, reducing the minimum local
$q_{0.05}$ to 0.8882. Importantly, the misspecified model simultaneously
produced more favorable ESS and maximum-weight diagnostics. Support
concentration metrics can therefore identify instability but cannot
validate the correctness of $g(A\mid U)$.

These results reinforce the distinction among four objects: the exact
graph-local target law, its sparse approximation, fitted nuisance models,
and finite-sample certification. Empirical agreement between exact and
sparse inference does not eliminate approximation error, just as close
oracle--estimated agreement does not make a fitted treatment model
theorem-certified. N2 and N3 keep these transitions explicit rather than
absorbing them into a single empirical-validity claim.

This distinction is particularly relevant in EO applications, where
flexible learning models can capture complex environmental structure while
formal characterization of dependence, treatment assignment, and
approximation error remains difficult
\cite{MaoMartinReich2024SpatialPrediction,
BarberEtAl2023BeyondExchangeability,
OliveiraEtAl2024Nonexchangeable}. RF, XGBoost, and fitted spatial
nuisance models can therefore be useful practical components of
GeoDose-CP without being assigned stronger guarantees than those actually
established.

\subsection{Earth-Observation Applicability, Limitations, and Outlook}
\label{subsec:discussion_eo}

The NSW study examines whether the inputs required by GeoDose-CP remain
scientifically defensible in a real EO archive. Its design combines
outcome-blind site selection, satellite-product quality screening,
pretreatment-only predictors, geographically separated validation, and
change-of-support analysis
\cite{RobertsEtAl2017SpatialCV,
MeyerEtAl2018TargetValidation,
KoldasbayevaEtAl2024GeospatialML}. It is therefore an applicability study,
not a real-data causal validation of full M6.

The principal boundary is treatment authenticity. The NSW archive contains
spatial rehabilitation information but not an auditable longitudinal
continuous rehabilitation treatment consistent with the estimand.
Constructing a snapshot-derived proxy would change the scientific question
rather than validate the intended intervention. Treatment-dependent
procedures therefore remain nonoperational, and the real-data analysis is
restricted to M1, M3, and the treatment-free graph-safe fallback GS. This
fail-closed behavior is central to the intended use of GeoDose-CP: the
availability of EO observations does not by itself justify an intervention
query.

These procedures operate on observed DEA photosynthetic-vegetation change,
which is a satellite-derived product rather than an error-free latent
measure of ecological recovery
\cite{Lymburner2021DEAFractionalCover,
MartinezFerrerEtAl2022BiophysicalUQ,
GarciaSoriaEtAl2024VegetationUQ}. Their coverage and width therefore
characterize observed-product predictive uncertainty, not causal
GeoDose-CP validity or uncertainty in latent ecological state.

The NSW results also show substantial geographic and model dependence.
Coverage differed markedly among Mt Arthur, HVO, and Bulga, and changing
from RF to XGBoost altered site-specific behavior rather than producing a
uniform shift. Pooled performance should therefore not be interpreted as
universal nominal coverage across mines. Consistently, 10 of 12
mine$\times$scale$\times$predictor spatial fits selected $\rho=0.8$, the
upper boundary of the fitted grid, indicating strong residual dependence
and/or spatial model mismatch. These fitted spatial nuisances, and their
associated N3 quantities, are empirical diagnostics rather than
finite-sample certified coverage bounds.

EO-product quality was not the dominant explanation for these patterns.
The stricter UE$\leq20$ sensitivity retained most eligible cells and did
not materially alter the principal site-, support-, or model-dependent
findings. In contrast, validation geometry had a large effect: random
partitioning produced substantially narrower intervals and lower prediction
error than geographically buffered validation. This is consistent with
optimistic assessment when training and test observations remain spatially
proximate
\cite{RobertsEtAl2017SpatialCV,
ValaviEtAl2019BlockCV,
PlotonEtAl2020SpatialValidation,
MeyerEtAl2018TargetValidation}. For geographically separated deployment,
validation design is therefore part of the scientific problem rather than
a purely technical choice.

Cross-mine transfer defines a still stronger support boundary. Under
leave-one-mine-out evaluation, the graph-dependent route could not form a
connected local calibration orbit for the held-out mine, so successful
cross-mine local-orbit transfer was not demonstrated
\cite{LudwigEtAl2023Transferability,
ZhangEtAl2026TransferConformal}. More generally, dense EO observation does
not guarantee intervention support: outcomes and covariates may be abundant
while the treatment provenance and spatial calibration structure required
for causal intervention inference remain absent.

Several limitations follow directly from these results. In the primary MineDoseBench benchmark, the unit residual scale makes the inverse-scale Jacobian equal to one and numerically inactive. GeoDose-CP can
produce wider sets in some coverage-matched regimes, and spatial/local
conformal methods such as E3 can be sharper. Sparse inference is close to
exact inference at aggregate and decision levels but can differ materially
for individual queries. Treatment-density misspecification can degrade
lower-tail local robustness even when weight diagnostics appear favorable,
and support gates do not eliminate false-support behavior or certify
untestable causal assumptions. The NSW study additionally lacks an
authentic longitudinal treatment, uses an observed EO product rather than
latent ecological truth, has no independent bare-ground confirmation,
provides no evidence of successful causal transfer to disconnected mines,
and does not establish finite-sample certification of the fitted real-data
spatial nuisance.

Future work should prioritize EO datasets containing an authentic
longitudinal rehabilitation treatment, repeated outcomes, and spatially
connected calibration support across sites. Methodologically, tighter
query-level sparse certificates, richer observational treatment models
including dependent assignment, and carefully bounded extensions to causal
interference are natural next steps
\cite{GiffinEtAl2023SpatialInterference,
PapadogeorgouEtAl2022SpatioTemporalCausal,
Leung2022SpatialInterferenceDesign}. These extensions should preserve the
central operating principle of GeoDose-CP: uncertainty about treatment
authenticity, spatial support, identification, approximation, or
certification should remain explicit rather than being converted into
unsupported extrapolation.

\section{Conclusion}
\label{sec:conclusion}

GeoDose-CP provides a support-aware framework for intervention-oriented EO
inference under continuous treatment and spatial dependence, centered on a
graph-local target-orbit law that jointly represents treatment shift,
outcome-scale transformation, and spatial residual dependence rather than
assuming naive factorization. Controlled experiments and MineDoseBench
provided complementary known-truth evidence; in MineDoseBench, M6 showed
strong selective and lower-tail local coverage across all 27 registered
configurations, while exact and sparse inference agreed closely at aggregate
and decision levels without implying pointwise equivalence. The fitted
treatment model closely tracked oracle behavior, whereas deliberate
misspecification weakened local robustness, and M6 showed the strongest
supported-regime coverage and lower-tail local profile on the selected
external-comparator subset, although spatial/local alternatives could be
sharper in coverage-matched regimes. The NSW study established a
complementary applicability boundary: without an authentic longitudinal
rehabilitation treatment, treatment-dependent inference remained
nonoperational rather than being forced through a proxy exposure. Overall,
GeoDose-CP links supported intervention targeting, graph-local conformal
inference, scalable approximation, and principled refusal while keeping
treatment authenticity and inferential support explicit.

\section*{Code and Data Availability}

The publication code, frozen experimental configurations, evaluation
scripts, MineDoseBench implementation and provenance materials, and
data-provenance records associated with this study are publicly available
at \url{https://github.com/Khalidsakib-121/GeoDose-CP}. The frozen
reproducibility release corresponding to this manuscript is available at
\url{https://github.com/Khalidsakib-121/GeoDose-CP/releases/tag/v1.0.0-paper}.

Derived materials that may be redistributed are included in the repository.
Original third-party data from NSW Government sources, Digital Earth
Australia, SILO, TERN, and Geoscience Australia are not redistributed;
source identifiers, access information, and processing provenance required
to reconstruct the analysis inputs are provided in the repository.

\section*{Declaration of Generative AI and AI-assisted Technologies in the Writing Process}

During the preparation of this work the authors used Grammarly in order to proofread the manuscript and improve the clarity of the English writing. After using this tool/service, the author(s) reviewed and edited the content as needed and takes full responsibility for the content of the published article.

\bibliographystyle{IEEEtran}

\bibliography{references}

\end{document}